\documentclass[a4paper,12pt]{article}

\usepackage{authblk}
\usepackage{hyperref}
\usepackage{braket}
\usepackage{amssymb}
\usepackage{amsmath}
\usepackage{amsthm}
\usepackage{amscd}
\usepackage[all,arc,knot,matrix]{xy}

\usepackage{graphicx}
\usepackage{geometry}
\usepackage{CJK}
\newtheorem{theorem}{Theorem}[section]
\newtheorem{lemma}[theorem]{Lemma}
\newtheorem{proposition}[theorem]{Proposition}
\newtheorem{corollary}[theorem]{Corollary}

\theoremstyle{definition}
\newtheorem{definition}[theorem]{Definition}
\newtheorem{example}[theorem]{Example}

\theoremstyle{remark}

\numberwithin{equation}{section}

\newcommand{\abs}[1]{\lvert#1\rvert}
\newcommand{\norm}[1]{\lVert#1\rVert}

\begin{document}
\begin{CJK*}{UTF8}{bsmi}
\title{ Thermofield Double States in Group Field Theory }

\author[1]{Xiao-Kan Guo (郭肖侃)\footnote {E-mail: kankuohsiao@whu.edu.cn }}
\affil[1]{Department of Physics, Beijing Normal University, Beijing 100875, China}

\date{\today}

\maketitle

\begin{abstract}
 Group field theories are higher-rank generalizations of matrix/tensor models, and encode the simplicial geometries of quantum gravity. In this paper, we study the thermofield double states in group field theories.
The starting point is the equilibrium Gibbs states in group field theory recently found by Kotecha and Oriti, based on which we construct the thermofield double state as a ``thermal" vacuum respecting the Kubo-Martin-Schwinger condition. We work with the Weyl $C^*$-algebra of group fields, and a particular type of thermofield double states with single type of symmetry are  obtained from the squeezed states on this Weyl algebra.  
The thermofield double states, when viewed as states on the  group field theory Fock vacuum, are condensate states at finite flow parameter $\beta$. We suggest that the equilibrium flow parameters $\beta$ of this type of thermofield double states in the group field theory condensate pictures of  black hole horizon and quantum cosmology are related to the inverse temperatures in gravitational thermodynamics.
\end{abstract}
\thispagestyle{empty}
\newpage
\pagenumbering{arabic}
\end{CJK*}
\section{Introduction}
Searching for the correct theory of quantum gravity is a recurring theme in theoretical physics. Many candidate theories have been proposed over the years, but no consensus is reached. In this respect, it is  worthwhile to find some common features from different theories of quantum gravity.
The group field theory (GFT) approach to quantum gravity can be related to many different approaches to quantum gravity, and hence is a natural place to find the common features of different theories.

GFTs are originally proposed to generate the simplicial quantum gravity such as the Ponzano-Regge model by topological lattice field theories with fields defined on the gauge group \cite{Bou92,Oog92}. It is realized in \cite{DFKR00} that the GFTs can also be utilized to generate the spin foam models and to assure the triangulation independence of spin foams. Since then  GFT provides an alternative way of viewing the spin foam models. In essence, the group fields are functions over the gauge group and hence incorporate the internal gauge symmetry quantum numbers into the fields, which is also the case for matrix models or  tensor models. The Feynman diagrams  in matrix or tensor models are simplicial graphs, and in GFT these simplical Feynman graphs can be related to the triangulations of simplicial quantum gravity models including spin foam models. Consequently the GFTs are natural background-independent field-theoretic models of quantum gravity based on discrete simplicial structures. See \cite{Kra11} for an introduction to the above aspects.

The field-theoretic structure  allows us to relate GFTs to many other approaches to quantum gravity in addition to simplicial quantum gravity: (i) when the gauge group is non-Abelian the group fields are noncommutative and can be formulated as a noncommutative field theory or noncommutative geometry \cite{FL06,BO10}; (ii) the simplicial Feynman graphs generated by GFTs can be interpreted as quantum geometric excitations, and the GFTs can be thus formulated as a second quantization of loop quantum gravity (LQG) built on the Ashtekar-Lewandowski vacuum where the excitations are the spin networks  \cite{Ori16}; (iii) the tensor-model structure and the spin-network structure can be combined into a tensor network representation of GFTs \cite{COZ18,CGOZ19} where a holographic duality can be built in analogy to random tensor networks; (iv) the algebra of creation and annihilation operators in the second quantization leads to an algebraic formulation of GFT \cite{KOT18}, thereby allowing the utilization of many techniques from algebraic quantum field theory. 

The second quantization formulation of GFT inspires many recent developments in GFT in relation to LQG. Basically, the second quantization formalism in quantum mechanics is designated to study the quantum many-body systems, and similarly the second quantization formalism of GFT or LQG is expected to describe the many-body physics, such as condensate states, of quantum spacetime ``atoms". See \cite{Ori07} for some early intuitions.  Remarkably, such a GFT condensate picture has been successfully applied to obtain a modified Friedmann equation in cosmology \cite{GOS13} and to explain the entropy of a quantum black hole \cite{OPS16}. A moment of reflection shows an obvious missing point in the GFT condensate picture of quantum gravity: there is no well-defined statistical mechanics for quantum gravitational states, since the Hamiltonian constraint of   pure gravity is not suitable for the conventional statistical-mechanical constructions, not to mention the lack of definitions of many thermodynamical quantities in quantum gravity. Nevertheless, recently in \cite{KO18} the equilibrium Gibbs states for GFT are constructed by alternative methods:  One of the methods is Jaynes' maximum entropy principle for constrained systems, which is further used in \cite{CKO19,Kot19} to study the generalized Gibbs states and background-independent statistical mechanics of quantum tetrahedra; the other method is based on the operator-algebraic Kubo-Martin-Schwinger (KMS) condition \cite{HHW67} and the algebraic formulation of GFT mentioned above. The Gibbs states in GFT are important for understanding the statistical mechanical aspects of quantum gravity in general, and they have advantages over other approaches to the equilibrium states in quantum gravity, such as \cite{AC16}, since in GFT  the collective ``many-body'' physics can be studied by many familiar field-theoretic techniques, which in turn helps the study of  classical limit for quantum gravity.

In this paper, we study the thermofield double (TFD) states in GFT based on the obtained Gibbs states of GFT. We take the operator-algebraic approach to Gibbs states, so that the algebraic TFD can be formulated  as in \cite{Oji81}. In particular, this algebraic approach to TFD uses Tomita-Takesaki modular theory to get the ``tilde" algebra $\tilde{\mathcal{M}}$ of TFD as the modular conjugate $\mathcal{M}'$ of the von Neumann algebra $\mathcal{M}$ of the original system. For a factor von Neumann algebra $\mathcal{M}$, $\mathcal{M}\neq\mathcal{M}'$ in general, but their  GNS representation Hilbert spaces are the same. 
The reason  is  that here the von Neumann algebra  $\mathcal{M}$ is constructed from the  full Fock states, instead of the local observable algebras in spacetime regions usually used in algebraic quantum field theories. By assuming the split property as in algebraic quantum field theory, we will see that the resulting intermediate type I factor can help us return to the conventional intuition of the TFD states in factorized Hilbert spaces \cite{Req13}.

The TFD states in GFT, when viewed as  states on the Fock vacuum, are in effect GFT condensate states. But  the TFD states now carry the equilibrium parameter $\beta$, the ``inverse temperature". Algebraically, the ``inverse temperature" $\beta$ also parametrizes the algebraic symmetry of the GFT algebraic states, which allows us to relate the ``thermal" behavior to the symmetry of quantum gravitational states, a link deeply implied by black hole thermodynamics.

We begin in Sec. \ref{s2} with the basic definitions of GFT, its algebraic formulation and the GFT Gibbs states. For the completeness of presentation, we include the essential proofs from \cite{KOT18,KO18}. In Sec. \ref{s3}, we define and study the TFD states of GFT using Tomita-Takeskai modular theory. By taking inspiration from the second-quantization interpretation of GFT, we first construct entangled squeezed states on the Weyl $C^*$-algebra of GFT and show that such algebraic squeezed states can be expressed in the standard form of TFD states \eqref{42} with a single set of generators. We then show  a graphic representation of  the obtained TFD states.
In Sec.\ref{s4}, we change to a diagrammatic presentation, and study the example of TFD shell condensate as a black hole horizon as well the sphere condensate for quantum cosmology to  discuss some speculative meanings of $\beta$ in this context. We summarizes this paper in Sec. \ref{s5}.
 In Appendix \ref{appA}, we  include another algebraic approach to the TFD extension of GFT based on the quantum deformation of the Hopf algebra of group fields. 
\section{Group field theories and Gibbs states}\label{s2}
Group fields are fields defined on $n$ copies of the group manifold of interest.
Formally GFTs are similar to tensor models, but the group fields have continuous indices, i.e. with infinite-dimensional tensor fields. The  physical input is the choice of the form of GFT action in which the simplicial structure of some models of quantum gravity can be encoded.

Let us work with the recent definitions given in \cite{CGOZ19}. Suppose the gauge group $G$ of interest is a locally compact Lie group with unimodular Haar measure $\mu$.
\begin{definition}
A group field $\phi$ is a $\mu^{\times n}$-integrable complex-valued function over the $n$-fold direct product of $G$,
\begin{equation}\label{1}
\phi: G^{\times n}\rightarrow\mathbb{C};\quad (g_1, g_2,..., g_n)\mapsto\phi(g_1,g_2,...,g_n)
\end{equation}
such that
\begin{equation}\label{2}
\int_{G^{\times n}}\phi(g_1,g_2,...,g_n)\overline{\phi(g_1,g_2,...,g_n)}d\mu^{\times n}<\infty.
\end{equation}
\end{definition}
Each group field $\phi$ can correspond to  a vector $\ket{\phi}$. These vectors  have inner products as in \eqref{2}
\begin{equation}
\braket{\phi|\phi'}=\int_{G^{\times n}}\phi(g_1,g_2,...,g_n)\overline{\phi'(g_1,g_2,...,g_n)}d\mu^{\times n}.
\end{equation}
By the Cauchy-Schwarz inequality and the condition \eqref{2}, this inner product $\braket{\phi|\phi'}$ is finite, so that the vectors $\ket{\phi}$ are in a Hilbert space $\mathcal{H}=L^2(G^{\times n},\mu^{\times n})$. The group fields $\phi$  can be recovered
 by assigning a linear functional
$\bra{g_1,...,g_n}$ in the dual space $\mathcal{H}^*$ to $\ket{\phi}$ such that\footnote{Notice that the group fields in this form $\braket{g_1,...,g_n|\phi}$ are single-``particle'' wave functions corresponding to the single-polyhedrons in the quantum geometries. The GFT field operators are constructed later in Sec.\ref{s2.1}.}
\begin{equation}\label{4}
\phi(g_1,...,g_n)=\braket{g_1,...,g_n|\phi}.
\end{equation}
The linear functionals $\bra{g_1,...,g_n}$ can be chosen to be multi-linear, that is, $\bra{g_1,...,g_n}=\bra{g_1}\otimes...\otimes\bra{g_n}$. In this case, the Hilbert space $\mathcal{H}$ can be factorized as
\begin{equation}\label{5}
\mathcal{H}=\bigotimes_{i=1}^n\mathcal{H}_i,\quad\mathcal{H}_i=L^2(G,\mu)_i
\end{equation}
where the vectors  $\ket{\phi_i}\in\mathcal{H}_i$ satisfy $\braket{g_i|\phi}=\braket{g_i|\phi_i}$. Notice that the factorization of $\ket{\phi}\in\mathcal{H}$ does not entail the factorization of the group field $\phi(g_1,...,g_n)$. In the special case where the group field $\phi$ is completely factorized, $\phi$ becomes the product $\prod_{i=1}^n\braket{g_i|\phi_i}$ of $n$  independent one-fold group fields. To avoid ambiguity, we assume in the following that the group fields $\phi$ cannot be factorized anymore.

To construct the action for GFT, consider first the change of  group elements in a group field $\phi(g_1,...,g_n)\rightarrow\phi(g'_1,...,g'_n)$ by the following transformation 
\begin{equation}
\phi(g'_1,...,g'_n)=\int_{G^{\times n}}\prod_{i=1}^nd\mu(g_i)C(g_1,...,g_n,g'_1,...,g'_n)\phi(g_1,...,g_n)
\end{equation}
where the integration kernel $C(g_1,...,g_n,g'_1,...,g'_n)$ is  the covariance. When expressed in terms of the Hilbert space representation \eqref{4}, the covariance $C$ is an endomorphism of the linear functionals, 
\begin{equation}
C:\bra{g_1,...,g_n}\mapsto\int_{G^{\times n}}\prod_{i=1}^nd\mu(g_i)C(g_1,...,g_n,g'_1,...,g'_n)\bra{g_1,...,g_n}.
\end{equation}
The integration kernel $C(g_1,...,g_n,g'_1,...,g'_n)$ thus encodes how the $g_i$ are transformed to $g'_i$.  
Then the kinematical term in the action can be formulated in analogy to the free  field theory:
\begin{definition}
Let $C$ be a covariance endomorphism on $\mathcal{H}^*$. A covariance $K$  is inverse to $C$ if $C\circ K={\bf1}_\mathcal{H}$. The  kinematical action of GFT is 
\begin{equation}\label{8}
S_0[\phi]=\frac{1}{2}\int_{G^{\times n}} \prod_{i=1}^nd\mu(g_i)\prod_{j=1}^nd\mu(g'_j)\phi(g_1,...,g_n)K(g_1,...,g_n,g'_1,...,g'_n)\overline{\phi(g'_1,...,g'_n)}.
\end{equation}
\end{definition}
With  \eqref{8}, the probability measure $D\phi\exp\{-S_0[\phi]\}$ becomes a multivariate Gaussian if $D\phi$ is chosen as  the  Lebesgue measure on $\mathbb{C}$. As in usual quantum field theories, an interaction term can be added to the free action, i.e. $S_0+\lambda S_{\text{int}}$ where $\lambda$ is a coupling constant. So far the GFT is defined in a general way. In the following, we restrict ourselves to the GFT that corresponds to the simplicial quantum gravity models:

Now the form of $S_{\text{int}}$ in GFT should be chosen in such a way that the Feynman diagrams obtained by perturbatively expanding the probability measure $D\phi\exp\{-S_0[\phi]-\lambda S_{\text{int}}[\phi]\}$ are dual to simplicial complexes of discrete quantum gravity.
\begin{definition}
Let $\phi_j(g_1,...,g_n)\equiv\phi(\{g_{i}^{(j)}\}),i=1,...,n;j=1,...,n+1$, be $(n+1)$ group fields defined on $G^{\times n}$. The interaction term in the action of GFT is
\begin{equation}\label{9}
S_{\text{int}}[\phi]=\frac{1}{n+1}\int_{G^{\times(n+1)}} \prod_{i\neq j=1}^{n+1}d\mu(g_i^{(j)})V(\{g_i^{(1)}\},...,\{g_i^{(n+1)}\}){\phi(g_i^{(1)})}...\phi(g_i^{(n+1)})
\end{equation}
where $V(\{g_i^{(1)}\},...,\{g_i^{(n+1)}\})$ is an integration kernel satisfying the closure constraints
\begin{equation}\label{10}
V(\{g_i^{(1)}\},...,\{g_i^{(n+1)}\})=\int_{G^{\times(n+1)}}\prod_{j=1}^{n+1}d\mu(h_j)V(\{h_1g_i^{(1)}\},...,\{h_{n+1}g_i^{(n+1)}\}),\quad\forall h_j\in G.
\end{equation}
\end{definition}
The Feynman diagrams in the leading order of the $\lambda$-expansion thus consist of the $n$-stranded graphs representing the free propagation, defined by $K$, of the $n$ group elements in the argument of $\phi(g_1,...,g_n)$. $(n+1)$ $n$-stranded graphs can meet at a vertex and the group fields get convoluted by $V$. The higher order terms in the $\lambda$-expansion then represent more complicated simplicial complexes.

Notice that the closure constraint \eqref{10} imposes gauge symmetry ($G$) at a given interaction vertex, which means only the gauge-invariant group fields contribute to the interactions in GFT. We can therefore  impose the global gauge invariance condition on group fields \eqref{1} and the kinematical kernel \eqref{8},
\begin{align}
\phi(g_1,...,g_n)=&\phi(hg_1,...,hg_n),\label{2.11}\\
K(\{g_i\},\{g'_i\})=&\int_{G^{\times2}}d\mu(h)d\mu(h')K(\{hg_i\},\{h'g'_i\}),
\end{align}
for $h,h'\in G$. The closure constraint \eqref{2.11} also  means that a single group field is dual to a polyhedron. To see this, let us take recourse from the noncommutative metric representation of group fields \cite{BO10}. A group field in the noncommutative metric representation  is obtained by the group Fourier transformation of the group fields, i.e.
\begin{equation}\label{2.13}
\hat{\phi}(\{x_i\})=\int\prod_i d\mu(g_i)\phi(\{g_i\})\prod_ie_{g_i}(x_i),\quad e_{g_i}(x_i)=e^{i\text{Tr}_{\mathfrak{g}}(x_i{g_i})},\quad g_i\in G,x_i\in\mathfrak{g}
\end{equation}
where the trace in the plane-wavefunction $e_g(x)$ is the trace on the Lie algebra $\mathfrak{g}$ of $G$ such that $\text{Tr}\tau_i\tau_j=-\delta_{ij}$ for generators $\tau_i$ of $\mathfrak{g}$. The Lie algebra elements $x_i$ can be associated to the vectors normal to the faces of a polyhedron in that $\abs{x_i}=\sqrt{\text{Tr}_{\mathfrak{g}}(x_i,x_i)}$ defines the area and $x_i/\abs{x_i}$ defines the unit normal vector of a face. In this representation the gauge-invariance of $\phi(\{g_i\})$ translates to the closure condition $\sum_ix_i=0$ of a polyhedron. Thus, the group fields convoluted by $V$ represent a 2-complexes containing these polyhedra.

The GFT thus defined can generate the amplitude of simplicial quantum gravity.
\begin{example}
In the Boulatov-Ooguri model \cite{Bou92,Oog92} over $G$, the choices of $K$ and $V$ are respectively
\begin{align}
K(\{g_i\},\{g'_i\})=&\int_{G}d\mu(h)\prod_{i=1}^n\delta(hg_i\bar{g}_i^{\prime}),\\
V(\{g_i^{(1)}\},...,\{g_i^{(n+1)}\})=&\int_{G}\prod_{i=1}^{n}d\mu(h_i)\prod_{i<j}\delta(h_ig_i^{(j)},h_j\bar{g}_i^{(j)}).
\end{align}
When $n=4$ and $G=SO(4)$, the convolution of $4$-stranded graphs at a vertex is dual to a 4-dimensional simplicial complex  (but not a polyhedron). The Feynman amplitude gives us the $15j$ symbol in the spin basis.
\end{example}
\begin{example}
Pithis {\it et al.}  \cite{PST16} choose the $K$ and $V$  as  
\begin{align}
K(\{g_i\},\{g'_i\})=&\prod_{i=1}^n\delta(g_i\bar{g}_i^{\prime})\Bigl[-\sum_i\Delta_{g_i}+m^2\Bigr],\\
V(\{g_i^{(1)}\},...,\{g_i^{(n+1)}\})=&\sum_{m=2}^{n+1}\frac{1}{n-1}{\phi^\dag(g_i^{(1)})}...\phi^\dag(g_i^{(m)})\prod_{m+1}^{n+1}\delta(\phi(g_i^{(m)}),1)\label{2.17}
\end{align}
where $\Delta_{g_I}$ is the Laplacian on the group manifold and the constant $m$ can be related to the spin foam edge weight. The $V$ in \eqref{2.17} give a tensorial nonlinear interaction term. 
\end{example}
The way in which $S_{\text{int}}$ enter the partition function implies that the interaction term is an observable in the free theory, i.e. $D\phi e^{-S_0}\lambda S_{\text{int}}$. Therefore general observables in GFT can be defined similarly as specific convolutions of group fields \cite{ORT15}:
\begin{definition}
A trace observable of GFT is a functional of the group fields with all the group elements are traced over, that is,
\begin{equation}\label{15}
O[\phi]=\int_{G^{J}} \prod_{i,j}d\mu(g_i^{(j)})B(\{g_i^{(j)}\})\prod_{i,j}{\phi(g_i^{(j)})}
\end{equation}
where  $j\in J$ not necessarily bounded by $(n+1)$ and $B$ is an integration kernel encoding the ways of convolution or tracing. A {\it partial} trace observable is a functional of the group fields with parts of the group elements are traced over, thereby being a functional $O_p[\phi,g]$ of both group fields and the untraced group elements.
\end{definition}
\begin{example}\label{2.6}
Consider the spin network states in LQG as GFT observables \cite{Ori16,Ori17}. The streamlined structure of a spin network graph $\gamma$ in LQG consists of a set $\mathcal{E}$ of edges colored with SU(2) spins $j$ and a set $\mathcal{V}$ of $n$-valent vertices to each of which is assigned an intertwiner operator $I$. Then $n$ colored edges are contracted with an $n$-valent vertex to form a spin network graph. Each vertex $v\in\mathcal{V}$ carries a quantum state $\ket{I_v}$ that represents a quantum polyhedron in the dual quantum geometry, and on each edge the spin-$j$ representation space can be chosen as a Hilbert space $\mathsf{H}_j$. Because $v$ is $n$-valent,  the intertwiner states $\ket{I_v}$  live in the Hilbert space 
\begin{equation}\label{16}
\mathsf{H}_{v}\equiv\bigotimes_{i=1}^n\mathsf{H}_{j_i,v}
\end{equation}
where $\mathsf{H}_{j_i,v}$ is the Hilbert space of the $i$-th spin-$j_i$ edge attached to the vertex $v$. Comparing \eqref{5} and \eqref{16}, we see that the intertwiner state   $\ket{I_v}$   corresponds to a GFT state $\ket{\phi}$. The difference is that $\ket{I_v}$ will be projected on a spin basis $\bra{\{j\}},j\in\mathbb{N}/2$  to obtain a spin network wavefunction, while $\ket{\phi}$ will be projected on a group basis$\bra{\{g_i\}},g_i\in G$ to obtain a group field.\footnote{Note that the group basis replaces the spin basis in the dual loop quantization based on Dittrich-Geiller vacuum with additional flatness constraints imposed on each closed loop or on the intertwiner states on each vertices \cite{DFG18}. In the case of GFT, these flatness constraints are not imposed, thereby covering also the intermediate ``squeezed" cases.
} 
 By writing $\ket{I_v}=\otimes_i\ket{j_i,v}$, we have the separate state $\otimes_v\otimes_i\ket{j_i,v}$ for all untraced vertices in $\mathcal{V}$. Next, a complete contraction of open edges can be made by matching the spins on a pair of open edges, which can be simply imposed on the spin basis as $\delta(j,j')\bra{j}\otimes\bra{j'}$. Therefore a spin network wavefunction on a closed graph $\gamma$ is
\begin{equation}\label{17}
\Psi_{\gamma,j,I_v}=\sum_{i=1}^n\sum_{\{j_i\}}M(j_i,j'_i)\delta(j_i,j'_i)\prod_{v\neq v'\in\mathcal{V}}\braket{j'_i,v'|j_i,v}
\end{equation}
where we have labelled the contracted open edges with the same $i$ and the $M$'s are elements of the representation matrix. Lifting the spin basis and intertwiner states respectively to the group basis and GFT states,\footnote{A traditional way of lifting is using the Peter-Weyl theorem to relate the group representation $\rho(g)$ to the discrete indices of the spin-$j$ representations.} 
we obtain the spin network observable in GFT in the form of \eqref{15} with the gauge-invariant gluing kernel
\begin{equation}
B(\{g_i^{(j)}\})=\int_{G}\prod_{i=1}^{n}d\mu(h_i)\prod_{i, j\neq j'}\delta(h_ig_i^{(j)},h_i\bar{g}_i^{(j')})M(h_ig_i^{(j)},h_i\bar{g}_i^{(j')})
\end{equation}
with some coefficients $M$.
When the spin network graph is not closed,  we have after contraction  spin network states on the remaining open edges. These ``boundary" spin network states have a tensor network representation \cite{COZ18},
\begin{equation}
\ket{\Psi_{\partial\gamma}}=\bigotimes_{e\in \mathcal{E}-\partial\gamma}\bra{M_e}\bigotimes_{v\in\mathcal{V}}\ket{I_v}
\end{equation}
where $\bra{M_e}$ represents  the contraction along the edge $e$.
\end{example}
\subsection{Algebraic formulation of group fields}\label{s2.1}
From Example \ref{2.6} we see that the observables in GFT are built from multiple group fields, so the second quantization formalism  becomes a natural language for studying these many GFT degrees of freedom. The resulting second quantization formulation is closely related to LQG \cite{Ori16,KOT18}. 

It is instructive to adopt the quantum geometric interpretation of a group field, namely a gauge-invariant group field, similar to a vertex in a spin network graph, corresponds to a dual quantum polyhedron. One  defines a GFT vacuum state $\ket{\Omega}$ without any quantum geometric excitation. Then a (gauge-invariant) group field $\phi$ and its conjugate $\phi^\dag$ respectively annihilates and creates a quantum polyhedron,
\begin{equation}\label{20}
\phi^\dag(\{g_i\})\ket{\Omega}=\ket{\{g_i\}},\quad \phi(\{g_i\})\ket{\Omega}=0.
\end{equation}
These creation and annihilation operators satisfy the {\it  bosonic} canonical commutation relations (CCR),
\begin{equation}
[\phi(\{g_i\}),\phi(\{g_j\})^\dag]=\delta(\{g_i\},\{g_j\}).
\end{equation}
These operators define a bosonic Fock space
\begin{equation}\label{22}
\mathcal{H}_{\text{Fock}}=\bigoplus_{N>0}\text{sym}\mathcal{H}^{\otimes N}
\end{equation}
where $\mathcal{H}$ is given in \eqref{5}. The occupation number basis, the field operators and in this Fock space and the many-body operators for the action can be formulated in the usual manner \cite{Ori16}.

To introduce the algebraic formulation of GFT, let us write down  the field operators in the group basis
\begin{equation}\label{23}
\Psi(\psi)=\int_{G^{\times n}}\prod_{i=1}^nd\mu(g_i) \overline{\psi(\{g_i\})} \phi(\{g_i\}),\quad \Psi(\psi)^\dag=\int_{G^{\times n}}\prod_{i=1}^nd\mu(g_i){ \psi(\{g_i\})} \phi(\{g_i\})^\dag
\end{equation}
where the $\psi(\{g_i\})$ are single-polyhedron wavefunctions. Namely, the group fields are field-operator-valued distributions on the space of ``single-particle" wavefunctions. These field operators $\Psi$ inherits the CCR algebra of $\phi$ with the $\delta$-function replaced by the  $L^2$ inner product of the $\psi$'s
\begin{equation}\label{24}
[\Psi(\psi),\Psi(\psi)^\dag]=\int_{G^{\times n}}\prod_{i=1}^nd\mu(g_i) \overline{\psi(\{g_i\}) }\psi(\{g'_i\})\equiv(\psi,{\psi}').
\end{equation}
Notice that the above CCR algebra is defined with respect to the single-body space $\mathcal{H}^{(1)}$, but it is the same as the CCR algebra obtained from the full Fock space \cite{Emc72}.
The inner product $(\psi,{\psi}')$ in \eqref{24} induces a symplectic form $\text{Im}(\psi,\psi')$ on the space of test functions $\psi$, and hence a Weyl algebra for GFT can be defined on the thus obtained phase space \cite{KOT18}. Here we recollect a more direct description:
\begin{proposition}
 Let $\Psi(\psi),\Psi(\psi)^\dag$  be the GFT field operators as in \eqref{23}. Then the exponentiated operators $W(\psi)=e^{\frac{i}{\sqrt{2}}(\Psi(\psi)+\Psi(\psi)^\dag)}$ form a  Weyl $C^*$-algebra.
\end{proposition}
\begin{proof}
Denoting $\Phi(\psi)=\frac{1}{\sqrt{2}}(\Psi(\psi)+\Psi(\psi)^\dag)$, we have $[\Phi(\psi_1),\Phi(\psi_2)]=i\text{Im}(\psi_1,\psi_2)$. Then by the Baker-Hausdorff formula, we see that $W(\psi)$ satisfy the defining relation for a Weyl algebra,
\begin{equation}
W(\psi_1)W(\psi_2)=e^{-\frac{i}{2}\text{Im}(\psi_1,\psi_2)}W(\psi_2+\psi_2).
\end{equation}
$W(\psi)$ is unitary, since $\Phi(\psi)$ is Hermitian. Hence $W(\psi)$'s form a Weyl algebra $\mathcal{W}$. The Hermitian conjugation defines the involution. The  $W$'s as bounded linear functionals can be represented on some Hilbert space $\mathcal{K}$, so we can assign a $C^*$-norm defined on the irreducible representations of $\mathcal{K}$,
\begin{equation}
\norm{W}_{C^*}=\sup_{\pi_{\mathcal{K}}}\norm{\pi_{\mathcal{K}}(W)}_{\mathcal{K}}=\sup_{\pi}\sqrt{(\pi_{\mathcal{K}}(W),\pi_{\mathcal{K}}(W))_\mathcal{K}},
\end{equation}
to $\mathcal{W}$, thereby making it a Weyl $C^*$-algebra.
\end{proof}
The Fock space structure of GFT can be recovered from the GNS representation $(\mathcal{H}_{\text{Fock}},\pi_F,\ket{\Omega})$ of $\mathcal{W}$. Here the GNS Hilbert space and the vacuum state are the same as the Fock space \eqref{22} and respectively the vacuum in \eqref{20} if the algebraic states on $\mathcal{W}$ are given by the quasi-free states
\begin{equation}\label{27}
\omega_F(W(\psi))=\braket{\Omega|\pi_F(W(\psi))|\Omega}=\braket{\Omega|e^{i\Phi(\psi)}|\Omega}=e^{-\frac{(\psi,\psi)}{4}}
\end{equation}
where we have used $\norm{\Phi(\psi)\ket{\Omega}}^2_{\mathcal{H}_{\text{Fock}}}=\frac{(\psi,\psi)}{2}$.
The important point is that the representation $\pi_F(W(\psi))=e^{i\Phi(\psi)}$ only  represents the {\it pure} states in $\mathcal{H}_{\text{Fock}}$ created by the field operator $\Phi$, thereby making the representation irreducible. The bicommutant of $\pi_F(W(\psi))$ is then the whole space $\pi_F^{\prime\prime}(W(\psi))=\mathcal{B}(\mathcal{H}_{\text{Fock}})$ of  bounded linear functionals, which is a von Neumann algebra. We thus have, alternative to the Weyl $C^*$-algebra $\mathcal{W}$, an  algebraic description of GFT with the von Neumann algebra $\mathcal{B}(\mathcal{H}_{\text{Fock}})$. 

We remark that working with the von Neumann algebra is  advantageous in that the interactions and dynamics can be studied algebraically through the Tomita-Takesaki modular theory \cite{ZW16}. Now that at a finite number $N$ of excitations in a subspace $\mathcal{H}_N$, $\mathcal{B}(\mathcal{H}_{N})=\bigoplus_{n=1}^N\mathcal{B}(\mathcal{H}_n)$ is a finite factor, there exist regular normal tracial states on $\mathcal{B}(\mathcal{H}_{N})$, whereby we have the statistical mixture of the Fock states: $\omega_\rho[W]=\text{Tr}(\rho\pi_F(W))$ where $\rho$ is a trace class density operator. For the infinite $N$ case, we can consider the  Weyl algebra $\mathcal{W}$ as a quasi-local algebra, i.e. the weak closure $\overline{\cup_N\mathcal{B}(\mathcal{H}_{N})}$, then  the normal states in $\mathcal{H}_{\text{Fock}}$ can be obtained by the $C^*$-inductive limit of the local normal states in the subspace $\mathcal{H}_N$, i.e. the local normal folium of states. In the following, the tracial states on $\mathcal{B}(\mathcal{H}_{\text{Fock}})$ should be understood in this sense.

\subsection{KMS condition and Gibbs states}
Working with the algebraic formulation of GFT as recalled in Sec.\ref{s2.1}, we can now apply the well-known algebraic theory of quantum statistical mechanics to GFT. 

In the algebraic formulation of a thermal  physical system, the equilibrium states are the  states satisfying the algebraic KMS condition \cite{HHW67}:
\begin{definition}
Let $\alpha_t$ be a one-parameter group of automorphisms of a $C^*$-algebra $\mathcal{A}$ (or a von Neumann algebra $\mathcal{M}$). Let $\omega[A],A\in\mathcal{A}$ be the algebraic states on $\mathcal{A}$. The KMS condition is
\begin{equation}
\omega[A\alpha_{t+i\beta}(B)]=\omega[\alpha_t(B)A], \quad t,\beta\in\mathbb{R}.
\end{equation}
\end{definition}
The algebraic states satisfying the KMS condition are KMS states.
The KMS  states are stationary with respect to the transformations or flows from a one-parameter group of algebraic automorphisms, i.e. $\omega[\alpha_t(A)]=\omega[A]$. Here $\beta$ is a flow parameter which is  not necessarily the inverse temperature.

In GFT, an existing one-parameter group of algebraic automorphisms consists of the left (or right) translations on $G$. Consider the collective left translations on $G^{\times n}$,
\begin{equation}
L_{\{g'_i\}}: G\rightarrow G;\quad (g_1,...,g_n)\mapsto(g'_1g_1,...,g'_ng_n)
\end{equation}
where the ``one-parameter" is the collective $\{g'_i\}=(g'_1,...,g'_n)$.
They can be pulled back to the space of ``single-polyhedron" wavefunctions $\psi(\{g_i\})$ as $L^*_{\{g'_i\}}\psi(\{g_i\})=\psi(L_{\{g^{\prime-1}_i\}}(\{g_i\}))$, thereby inducing a flow on the Weyl elements
\begin{equation}\label{29}
\alpha_{\{g'_i\}}:\mathcal{W}\rightarrow\mathcal{W};\quad W(\psi)\mapsto\alpha_{\{g'_i\}}(W(\psi))=W\bigl(L^*_{\{g'_i\}}\psi(\{g_i\})\bigr)
\end{equation}
which is a *-automorphism on $\mathcal{W}$. Because the von Neumann algebra of interest is the whole $\mathcal{B}(\mathcal{H}_{\text{Fock}})$, the *-automorphisms $\alpha_{\{g_i\}}$ can be represented as unitary transformations $U(g),g\in G$, on $\mathcal{H}_{\text{Fock}}$   by the bounded linear transformation theorem,
\begin{equation}\label{30}
U(\{g'\})\Psi(\{g_i\})U(\{g'\})^{-1}=\Psi(\{g'_ig_i\}).
\end{equation}
Such a unitary operator $U$ admits a Hertmitian generator $\mathcal{G}$ such that $U(t)=e^{it\mathcal{G}}$ for some parameter $t\in\mathbb{R}$. Not surprisingly, this $\mathcal{G}$ gives rise to the canonical Gibbs form of KMS states \cite{KO18}:
\begin{lemma}\label{2.9}
Let $\pi_F(W)$ be the Fock GNS representation map for $W\in\mathcal{W}$ and $\omega_\rho[W]=\text{Tr}(\rho \pi_F(W))$ be a mixed algebraic state on $\mathcal{W}$ with a density operator $\rho$. If the $\omega_\rho[W]$'s are KMS states with respect to an one-parameter group of automorphisms  $\alpha_{t}$, then $\rho\propto e^{-\beta\mathcal{G}}$ for some $\beta\in\mathbb{R}$.
\end{lemma}
\begin{proof}
The KMS condition for  $\omega_\rho[W]$ is $\omega_\rho[W\alpha_{t+i\beta}(W')]=\omega_\rho[\alpha_{t}(W')W]$. When $t=0$, this becomes
\[\text{Tr}(\rho\pi_F(W')\pi_F(W))=\text{Tr}(\rho\pi_F(W)\pi_F(\alpha_{i\beta}W')).
\]
Using $W=e^{i\Phi}, U=e^{it\mathcal{G}}$ and \eqref{30}, we have
\[\text{Tr}(\rho\pi_F(W')\pi_F(W))=\text{Tr}(\rho\pi_F(W)e^{-\beta\mathcal{G}}\pi_F(W')e^{\beta\mathcal{G}})=\text{Tr}(e^{-\beta\mathcal{G}}\pi_F(W')e^{\beta\mathcal{G}}\rho\pi_F(W)).
\]
Therefore $\rho\pi_F(W')=e^{-\beta\mathcal{G}}\pi_F(W')e^{\beta\mathcal{G}}\rho$, which entails $[e^{\beta\mathcal{G}}\rho,\pi_F(W')]=0$ or $e^{\beta\mathcal{G}}\rho\in\pi'_F(\mathcal{W})$. Since $\pi_F^{\prime\prime}(W(\psi))=\mathcal{B}(\mathcal{H}_{\text{Fock}})$  and $\pi_F$ is irreducible, we have  that $\pi'_F(\mathcal{W})\propto{\bf1}$. Hence $\rho\propto e^{-\beta\mathcal{G}}$.
\end{proof}
Next, to apply the above result to the collective left translations \eqref{29},  we suppose that $G$ is connected, which holds  for all groups used in GFT models related to spin foams and simplicial quantum
gravity \cite{KO18}:
\begin{lemma}\label{l2.10}
Let $G$ be a connected Lie group and $\mathfrak{g}$ be its Lie algebra. The KMS states  (i.e. density operators) $\rho_X$ with respect to the collective left translations $\alpha_{\{g'_i\}},g'_i\in G$  have the canonical Gibbs form
\begin{equation}\label{g}
 \rho_X=\frac{1}{Z}e^{-\beta\mathcal{G}_X},\quad Z=\text{Tr}e^{-\beta\mathcal{G}_X},\mathcal{G}_X=iU_*(X),\quad\beta\in\mathbb{R}, X\in\mathfrak{g}
\end{equation}
where $U_*$ is the anti-Hermitian representation of $\mathfrak{g}$.
\end{lemma}
\begin{proof}
Consider $X\in\mathfrak{g}$, then $X$ can be mapped to a $g_X(t)=e^{tX}\in G$ through the exponential map, where $t\in\mathbb{R}$ is a parameter such that $g_X(0)=1_G, (dg_X/dt)|_{t=0}=X$. The left translations on G become $L_{g_X}=L_{e^{tX}}=g_X(\cdot)$. Let $U:G\rightarrow \mathcal{U}_\mathcal{{H}'}$ be a strongly continuous unitary representation of $G$ by the unitary operators on some Hilbert space $\mathcal{{H}'}$, then the composition $U_X=U\circ g_X$ is a strongly continuous one-parameter ($t$) group of unitary operators in $\mathcal{U}_\mathcal{{H}'}$. Hence $U_X$ is a strongly continuous unitary representation of $\alpha_g$.
By Stone's theorem, we have $U_X=e^{-i\mathcal{G}_Xt}$ for some Hermitian generator $\mathcal{G}_X$ from  $\mathcal{{H}'}$. Applying Lemma \ref{2.9} then gives the Gibbs states. Finally, the  anti-Hermitian $U_*(X)$ comes from $U_X(t)=U(e^{tX})=e^{t(U_*(X))}=e^{-it\mathcal{G}_X}$.
\end{proof}
Finally, in the case of GFT, we have $\mathcal{{H}'}=\mathcal{H}_{\text{Fock}}$ and the connected group is $G^{\times n}$. Notice that the collective left translations on $G$, when written as $g_{X_i}(\cdot)=e^{tX_i}(\cdot)$,  have right-invariant generators $X$, otherwise the (global) gauge invariance will be violated when acting different $g_i$. We therefore have
\begin{theorem}[Kotecha-Oriti \cite{KO18}]
Let $\mathcal{H}_{\text{Fock}}$ be the  Fock space for a GFT on $G^{\times n}$, and let $\alpha_{\{g_i\}}$ be the collective left translations on $G^{\times n}$, then the KMS states $\rho_g$ with respect to $\alpha_{\{g_i\}}$ take the canonical Gibbs form $\rho_g=e^{-i\beta U_*(X,g)}/Z, \beta\in\mathbb{R}$ where
\begin{equation}\label{u}
U_*(X,g)=\int_{G^{\times n}}\prod_{i=1}^nd\mu(g_i) \overline{\phi(\{g_i\})}\mathcal{L}_{R_{g^*}X} \phi(\{g_i\})
\end{equation}
with $\mathcal{L}_{R_{g^*}X}$ being the Lie derivative along the right-invariant generators $R_{g^*}X$ of $L_{g_X}$ for $X\in\mathfrak{g}$.
\end{theorem}
\section{Thermo field dynamics  via Tomita-Takesaki theory}\label{s3}
 TFD is a reformulation of the real time formalism, i.e. the closed time path formalism, of quantum field theories at finite temperature \cite{Das97}. Basically, TFD intends to rewrite the partition function of a  thermal  quantum system  in the form of a vacuum-vacuum transition amplitude as in a zero-temperature field theory. For a state $\rho$ in equilibrium at inverse temperature $\beta_T=1/T$, the thermal partition function is $Z=\text{Tr}(e^{-\beta_T H})$, the trace of the Gibbs state with some Hamiltonian $H$. Then the thermal  average of an operator $A$ can be written as 
\begin{equation}
\braket{A}_{\beta_T}=\frac{1}{Z}\text{Tr}(Ae^{-\beta_TH})\equiv\braket{0,\beta_T|A|0,\beta_T}
\end{equation}
where the thermal vacuum $\ket{0,\beta_T}$, when written in the energy eigenbasis $\{\ket{n}\}$, is
\begin{equation}\label{34}
\ket{0,\beta_T}=\frac{1}{\sqrt{Z}}\sum_ne^{-\frac{1}{2}\beta_T E_n}\ket{n}\otimes\ket{\tilde{n}}.
\end{equation}
The tensor product basis states $\ket{n}\otimes\ket{\tilde{n}}$ are in the doubled Hilbert space $\mathcal{H}\otimes\mathcal{\tilde{H}}$ where $\mathcal{H}$ is the quantum system of interest and $\mathcal{\tilde{H}}$ is a copy of $\mathcal{H}$ so as to produce the same $\delta$-functions, i.e. $\delta_{mn}=\braket{\tilde{m}|\tilde{n}}=\braket{m|n}$. In the formalism of second quantization, the coefficients $e^{-\frac{1}{2}\beta_T E_n}$ comes from the Bogoliubov transformation between the two sets of  annihilation and  creation  operators:  $(a,a^\dag)$ for  $\mathcal{H}$ and $(\tilde{a},\tilde{a}^\dag)$ for  $\mathcal{\tilde{H}}$.
The time evolutions in  $\mathcal{H}\otimes\mathcal{\tilde{H}}$ are then generated by the free total Hamiltonian $H-\tilde{H}$, and this total Hamiltonian annihilates the thermal vacuum, $(H-\tilde{H})\ket{0,\beta}=0$. Interaction terms $\lambda(L_I-\tilde{L}_I)$ can be added to the free total Lagrangian $L_0-\tilde{L}_0$, so that the perturbative Feynman diagrams can be calculated for the doubled system.

The $\tilde{\mathcal{H}}$ in TFD is usually considered as a ficticious tool to express the thermal mixed state in $\mathcal{H}$ as a pure entangled state in $\mathcal{H}\otimes\mathcal{\tilde{H}}$. However, since the thermal vacuum is entangled, $\tilde{\mathcal{H}}$  surely contains the physical information of $\mathcal{H}$. 
In the following we shall investigate  the quantum geometric meaning  of $\tilde{\mathcal{H}}$ in the TFD extension of GFT.
\subsection{Algebraic thermofield dynamics}\label{s3.1}
The  thermal vacuum \eqref{34} of TFD is based on a well-defined Hamiltonian $H$ of a thermal system generating the time evolutions. For GFT and the related theories of quantum gravity, the Hamiltonian might not exist, but the operator-algebraic formulation is still valid.

Let $\mathcal{A}$ be a $C^*$-algebra of physical observables with faithful algebraic  states $\omega$, then we have the GNS representation $(\mathcal{H}_\omega,\pi_\omega,\Omega_\omega)$ of $\mathcal{A}$ constructed from $\omega$. Let $\mathcal{M}=\pi_\omega^{\prime\prime}(\mathcal{A})$ be the von Neumann algebra generated by $\pi_\omega(\mathcal{A})$, then for $M\in\mathcal{M}$ it allows an antilinear operator $S$ on $\mathcal{H}_\omega$ such that $SM\Omega_\omega=M^\dag\Omega_\omega$, or equivalently $S\ket{A}=\ket{A^\dag}$ for $\ket{A}\in\mathcal{H}_\omega$.
The  polar decomposition $S=J\Delta^{1/2}$ defines  the modular conjugation operator $J=J^\dag$ and  the modular operator $\Delta=S^\dag S$. Then the Tomita-Takesaki modular theory of von Neumann algebra \cite{ZW16}  tells us that
\begin{equation}
J\mathcal{M}J=\mathcal{M}',\quad \Delta^{it}\mathcal{M}\Delta^{-it}=\mathcal{M},\quad t\in\mathbb{R}
\end{equation}
where the former is the commutation relation and the latter dictates the ``time evolution", or a a flow $\Delta^{it}(\cdot)\Delta^{-it}\equiv\sigma_{t}$ of the algebra without referring to an explicit Hamiltonian.
 The Tomita-Takesaki theoretic formulation of TFD is as follows.
\begin{theorem}[Ojima \cite{Oji81}]
For a thermal quantum system described by a von Neumann algebra $\mathcal{M}$ constructed from the KMS states $\omega_{\text{KMS}}$ with respect to the flow generated by the modular operator $\sigma_{t}$, the TFD of $\mathcal{M}$ consists of the ``tilde" system described by $\tilde{\mathcal{M}}=J\mathcal{M}J=\mathcal{M}'$ with $J$ being the modular conjugate operator and  the thermal vacuum $\Omega_{\omega_{\text{KMS}}}$ which is the GNS vacuum of $\omega_{\text{KMS}}$.
\end{theorem}
\begin{proof}
We first note that the flow $\sigma_{t}$ can be considered as a flow with generator $\bar{H}$ such that $\Delta^{it}(\cdot)\Delta^{-it}=e^{i\bar{H}t}(\cdot)e^{-i\bar{H}t}$ for $\Delta=e^{-\beta\bar{H}}$ and $\beta=-1$. Given the GNS vacuum $\Omega_{\omega_{\text{KMS}}}$, one has $\Delta^{it}\Omega_{\omega_{\text{KMS}}}=0$, which implies $\bar{H}\Omega_{\omega_{\text{KMS}}}=0$. Since $J\Delta J=\Delta^{-1}$, one has furthermore $J\bar{H}J=-\bar{H}$. If the system described by $\mathcal{M}$ has a well-defined Hamiltonian $H$, then $\bar{H}$ admits the expression $\bar{H}=H-JHJ$, which indicates that the ``tilde" space should be obtained by the modular conjugation.
To proceed without a Hamiltonian, we suppose $\bar{H}$ is merely a generator of the flow $\sigma_t,t\in\mathbb{R}$. Then we have
\[M^\dag\Omega_{\omega_{\text{KMS}}}=SM\Omega_{\omega_{\text{KMS}}}=J\Delta^{-1/2}M\Omega_{\omega_{\text{KMS}}}=Je^{-\beta\bar{H}/2}M\Omega_{\omega_{\text{KMS}}},\]
which corresponds to the thermal condition underlying the KMS condition in TFD \cite{Das97} with only one additional $J$. From the Tomita-Takesaki theory, we also have the properties that $J\Omega_{\omega_{\text{KMS}}}=\Omega_{\omega_{\text{KMS}}}$ and $JMJ\in\mathcal{M}'$. So by identifying $JMJ=\tilde{M}\in\tilde{\mathcal{M}}$, we obtain $\braket{A}_\beta=(\Omega_{\omega_{\text{KMS}}},A\Omega_{\omega_{\text{KMS}}})$ satisfying the KMS condition.
\end{proof}
Notice that in general $\mathcal{M}'$ is not the same as $\mathcal{M}$ (because we want $\mathcal{M}$ to be a factor),  but on the level of elements $M\in{\mathcal{M}}$ and $JMJ\in{\mathcal{M}'}$, so the corresponding Fock-space structures are the same. The reason is that for the GNS representation $(\mathcal{H}_\omega,\pi_\omega,\Omega_\omega)$ of $\mathcal{A}$ constructed from a state $\omega$ and the von Neumann algebra $\mathcal{M}=\pi_\omega^{\prime\prime}(\mathcal{A})$, since the vacuum $\Omega_\omega$ is cyclic and separating, we have $\mathcal{H}_\omega=\overline{\pi_\omega\Omega_\omega}=\overline{\mathcal{M}\Omega_\omega}=\overline{\mathcal{M}'\Omega_\omega}$  where the overline denotes the norm closure. In this sense, the above algebraic TFD is more general than the Hilbert-space version.

To retain the description via factorized Hilbert spaces in the algebraic formulation, we can assume the split property as for the net of algebras in relativistic quantum field theory \cite{DL84,Req13}. Consider two von Neumann algebras $\mathcal{M}_1,\mathcal{M}_2$ of  observables with $\mathcal{M}_1\subset\mathcal{M}_2$ (or $\mathcal{M}_1\cap\mathcal{M}'_2=\emptyset$), then by the split property there is an intermediate type I factor $\mathcal{M}$ such that $\mathcal{M}_1\subset\mathcal{M}\subset\mathcal{M}_2$. For  $\mathcal{M}_1\lor\mathcal{M}_2'(\cong\mathcal{M}_1\otimes\mathcal{M}_2')$, there is a cyclic and separating vector $\eta\in\mathcal{H}$ \footnote{ $\mathcal{H}$ is the separating Hilbert space on which  the net of von Neumann algebras  $\mathcal{M}_i$ act.} and a unitary operator $W:\mathcal{H}\rightarrow\mathcal{H}\otimes\mathcal{H}$ such that $W^*AB'W=A\otimes B'$ and $W\eta=\Omega\otimes\Omega$. Then  $\mathcal{M}=W^*(\mathcal{B}(\mathcal{H}\otimes1))W,\mathcal{M}'=W^*(\mathcal{B}(1\otimes\mathcal{H}))W$, and therefore $\overline{\mathcal{M}\lor\mathcal{M}'}=\mathcal{B}(\mathcal{H})$. The localized finite-dimensional Hilbert space is $\mathcal{H}_\mathcal{M}=\mathcal{M}\eta$, since $\mathcal{H}_\mathcal{M}=W^*\mathcal{H}\otimes\Omega=W^*(\mathcal{B}(\mathcal{H})\otimes1)WW^*(\Omega\otimes\Omega)$.
TFD can be formulated in terms of the localized finite-dimensional $\mathcal{H}_\mathcal{M},\mathcal{H}_\mathcal{M}'$:
\begin{theorem}[Requardt \cite{Req13}]\label{req}
Let $\mathcal{M}$  be the intermediate factors as above. Then the TFD of $\mathcal{M}$ consists of the ``tilde'' system $\mathcal{M}'$ and the following modular data on the factorized localized Hilbert spaces $\mathcal{H}_\mathcal{M}\otimes\mathcal{H}_{\mathcal{M}'}$ satisfying KMS condition: the modular operator 
$
\hat{\Delta}_\Omega=\hat{\rho}_\Omega\otimes\hat{\tilde{\rho}}^{-1}_\Omega
$ on $\mathcal{M}$, the modular conjugate $J$ acting as $J(\psi_i\otimes\psi_j')=\psi_j\otimes\psi_i'$, where $\hat{\rho}_\Omega$ are the density operators on $\mathcal{H}_\mathcal{M}$ and the $\psi_i$ are the eigenstates of $\hat{\rho}_\Omega$, and likewise for $\mathcal{M}'$.
\end{theorem}
\begin{proof}
Consider density operator $\rho,\tilde{\rho}$ on $\mathcal{H}$, then $\rho\otimes\tilde{\rho}^{-1}\equiv\Delta$ defines the modular operator on $\mathcal{H}\otimes\mathcal{H}$. The modular conjugate $J$ acts as $J(e_i\otimes f_j)=e_j\otimes f_i$ for the orthonormal basis $\{e_i\otimes f_j\}$ of  $\mathcal{H}\otimes\mathcal{H}$. It is easy to check that $\Delta^{is}(\cdot)\Delta^{-is}$ generates a modular flow on $\mathcal{B}(\mathcal{H})\otimes1$, and likewise  $\Delta^{-is}(\cdot)\Delta^{is}$ generates a modular flow on $1\otimes\mathcal{B}(\mathcal{H})$; the vector $\hat{\Omega}=W\Omega$  satisfies the KMS condition under both flows. Next, using $W$ we can transform the KMS conditions to $\mathcal{M}$ and $\mathcal{M}'$ with respect to $\Omega$ with 
\[\hat{\Delta}=\hat{\rho}\otimes\hat{\tilde{\rho}}^{-1}=W^*(\rho\otimes1)W\cdot W^*(1\otimes\tilde{\rho}^{-1})W\]
on $\mathcal{M}$. This can be furthermore localized to $\mathcal{H}_\mathcal{M}$ by noting the localization of invertibel density operators: 
$\hat{\rho}_\Omega=W^*(\rho\otimes P_\Omega)W$ where $P_\Omega$ is a projector onto $\Omega$. Then for $A,B\in \mathcal{M}$, we have the KMS condition 
\[\text{tr}[\hat{\rho}_\Omega(\hat{\rho}_\Omega^{is}A\hat{\rho}_\Omega^{-is})B]=\text{tr}[B(\hat{\rho}_\Omega^{i(s-i)}A\hat{\rho}_\Omega^{-is})]=\text{tr}[\hat{\rho}_\Omega B(\hat{\rho}_\Omega^{i(s-i)}A\hat{\rho}_\Omega^{-i(s-i)})B].
\]
The similar holds on $\mathcal{M}'$.
\end{proof}

\subsection{Thermofield double states in group field theories }
Let us first recall the thermal vacuum of a free bosonic oscillator with Hamiltonian $H_b=\epsilon a^\dag a$ \cite{Das97},
\begin{equation}\label{37}
\ket{0,\beta}_b=\sqrt{1-e^{-\beta\epsilon}}\sum_{n=0}^\infty e^{-n\beta\epsilon/2}\ket{n}\ket{\tilde{n}}=e^{\theta(\beta)(a^\dag\tilde{a}^\dag-\tilde{a}a)}\ket{0}\ket{\tilde{0}}
\end{equation}
where $(\tilde{a},\tilde{a}^\dag)$ are the annihilation and creation operators on $\tilde{\mathcal{H}}$, and the parameter $\theta(\beta)$ is defined by the Bogoliubov transformation coefficients
\begin{equation}\label{bogo}
\cosh\theta=(1-e^{-\beta\epsilon})^{-1/2},\quad \sinh\theta=e^{-\beta\epsilon/2}(1-e^{-\beta\epsilon})^{-1/2}.
\end{equation}
The thermal vacuum \eqref{37} is a squeezed state with the temperature-dependent squeezing parameter $\theta(\beta)$. Now to formulate the TFD of GFT, we first need to perform the GNS construction from the Weyl algebra $\mathcal{W}$ and a new vacuum $\ket{\Omega_S}$ different from the Fock vacuum $\ket{\Omega_F}=\ket{\Omega}$. From \eqref{37}, we expect the new vacuum $\ket{\Omega_S}$ to define a $\beta$-dependent squeezed state representation similar to the coherent state representation  \cite{KOT18}.
The squeezed states on a Weyl $C^*$-algebra have been studied, for example, in \cite{HR96,HR97}. We similarly define the  GFT squeezed states on $\mathcal{W}$:
\begin{definition}
Let $\mathcal{W}$ be the Weyl $C^*$-algebra for GFT with Weyl elements $W(\psi)$. A squeezed state on $\mathcal{W}$ is
\begin{equation}\label{38}
\omega_S(W(\psi))=\omega_{F}(W(\psi))e^{-\frac{1}{2}\text{Re}\{c(\psi,\bar{\psi})\}},\quad c\in\mathbb{C}
\end{equation}
where $\omega_{F}(W(\psi))$ is the Fock state \eqref{27}.
\end{definition}
The extra term added to the Fock state can be understood as   coming from the variances (or fluctuations) of the  self-adjoint field operators $(\Phi,\Phi^\dag)$ of $(\Psi,\Psi^\dag)$ in the expectation values of a squeezed state $\ket{\Omega_S}$. By requiring the special relations
\begin{align}
\braket{\Omega_S|\Psi(\psi)\Psi(\psi')|\Omega_S}=&c(\psi,\bar{\psi}'),\quad\braket{\Omega_S|\Psi(\psi)|\Omega_S}=0,\nonumber\\
 \braket{\Omega_S|\Psi^\dag(\psi)\Psi^\dag(\psi')|\Omega_S}=&\bar{c}(\bar{\psi},{\psi}'),\quad \braket{\Omega_S|\Psi^\dag(\psi)|\Omega_S}=0,\label{3.7}
\end{align}
where $c\in\mathbb{C}$. We see that the  the terms in \eqref{3.7} contribute to the covariance as $\text{Re}\{c(\psi,\bar{\psi})\}$.
We can   write the squeezed vacuum $\ket{\Omega_S}$ in a Fock space in the following way.  Consider the Fock space for GFT with the Fock  vacuum state $\ket{\Omega}$, then
 the (single mode) squeezed state can be written as the usual form of squeezed vacuum on Fock space 
\begin{equation}\label{40}
\ket{\Omega_S}=e^{\frac{b}{2}\Psi^\dag(\psi)\Psi^\dag(\psi')+\frac{\bar{b}}{2}\Psi(\psi)\Psi(\psi')}\ket{\Omega}, 
\end{equation}
where $b\in\mathbb{C}$ is the squeezing parameter. When Re$\{c(\psi,\bar{\psi})\}=0$, the state \eqref{38} reduces to the Fock state without squeezing.

We are interested in the von Neumann algebra $\mathcal{M}\otimes\mathcal{M}'\equiv\mathsf{M}$ of the TFD type. 
So let $\mathsf{W}=\mathcal{W}\otimes\tilde{\mathcal{W}}$ be the Weyl $C^*$-algebra corresponding to $\mathsf{M}$ and $\omega_{\mathsf{W}}:\mathsf{W}\rightarrow\mathbb{C}$ be the algebraic states on $\mathsf{W}$. 
By the same reasoning of Sec. \ref{s2.1}, we can construct from $\omega_{\mathsf{W}}$ the GNS triple $(\mathcal{H}_{\mathsf{W}},\pi_{\mathsf{W}},\ket{\Omega_{\mathsf{W}}})$. If this GNS representation restricted to $\mathcal{W}$ is pure, we have $\tilde{\mathcal{M}}=J\mathcal{M}J=\mathcal{M}'=\eta{\bf1}$ with $\eta$ a constant. Let us denote $\mathsf{M}_0=\mathcal{M}\otimes{\bf1}$.  When the total GNS representation on $\mathsf{W}$ is Fock, we still have the pure states $\mathcal{\omega}_{\mathsf{W},F}$. 
 If furthermore a $\mathcal{\omega}_{\mathsf{W},F}$ is separable and satisfy $\mathcal{\omega}_{\mathsf{W},F}|_\mathcal{W}=\omega_F$, then
$\mathcal{\omega}_{\mathsf{W},F}\in\overline{\text{conv}}(\omega_F\otimes {\bf1})$, the norm closure of convex combination of product Fock states.
On the other hand, if $\omega_{\mathsf{W}}=\omega_{\mathsf{W},\text{KMS}}$ is the KMS states on $\mathsf{W}$  with respect to the left translations on the group-element arguments  (which is not affected by $J$) in $\psi$, then $\tilde{\mathcal{M}}\neq\eta{\bf1}$. In this case, $\mathsf{M}_0$ becomes a subalgebra of $\mathsf{M}$.
We can also define the (two-mode) squeezed states $\omega_{\mathsf{W},S}$ on $\mathsf{W}$ with the replacement $\Psi(\psi')\rightarrow\tilde{\Psi}(\psi')$. In this case, the relations in \eqref{3.7} can be violated, but the definition of squeezed vacuum state still holds. 

For $\mathsf{M}$ to describe an algebraic TFD, it is required that $\ket{\Omega_\mathsf{W}}=\ket{\Omega_S}=\ket{\Omega_{\omega_{\text{KMS}}}}$.
The following result shows that $\omega_{\mathsf{W},S}$ is qualified to be a TFD state, i.e. an entangled squeezed states.
\begin{proposition}
Let  $\mathsf{W}=\mathcal{W}\otimes{\mathcal{W}}'$ be a Weyl $C^*$-algebra. Then the squeezed states $\omega_{\mathsf{W},S}$ on $\mathsf{W}$ is an entangled state.
\end{proposition}
\begin{proof}
Consider the Fock space representation of pure states $\omega_{\mathsf{W},F}$. Then any state on $\mathsf{W}$ can be written as  convex combination $\omega_{\mathsf{W}}=\sum_i\lambda_i\omega_{\mathsf{W},F,(i)}$ with  $\sum_i\lambda_i=1$. If the squeezed states $\omega_{\mathsf{W},S}$ on $\mathsf{W}$  is separable, then we must have $\mathcal{\omega}_{\mathsf{W},S}\in\overline{\text{conv}}(\omega_F\otimes {\bf1})$. Since the $\omega_{\mathsf{W},S}$ are not Fock states, so  $\mathcal{W}'\neq\eta{\bf1}$ and $\mathcal{\omega}_{\mathsf{W},S}\notin\overline{\text{conv}}(\omega_F\otimes {\bf1})$. Hence the  $\omega_{\mathsf{W},S}$ are  entangled.

Alternatively, we  would like to avoid using  $\omega_{\mathsf{W},F}$ which has not been explicitly constructed. Let us consider instead the subalgebra $\mathsf{W}_0$ corresponding to $\mathsf{M}_0$. Then the restriction of the domain of $\omega_{\mathsf{W},S}$ to $\mathsf{W}_0$ is the squeezed states $\omega_S$ on $\mathcal{W}_F$, and it has the expression $\omega_S[W]=\text{Tr}(\rho_S\pi_F(W))$ for $W\in\mathcal{W}$ in the sense of local normal folium.
By the definition of $\omega_S$ \eqref{38}, we have that $\rho_S$ is reducible and has the convex decomposition $\rho_S=\sum_n\lambda_n\rho_F^{(n)}$ where $n$ labels the particle number in the Fock space.
By the algebraic approach \cite{BGdR13}, the von Neumann entropy between $\omega_{\mathsf{W}_0,S}$ and  $\omega_{\mathsf{W}-\mathsf{W}_0,S}$ is $S=-\sum_n\lambda_n\ln\lambda_n\neq0$. Therefore  $\omega_{\mathsf{W},S}$  is  entangled.
\end{proof}

Next, we need to find the Bogoliubov transformation that relates the squeezed vacuum states to the KMS  states and find a way to return to the Fock space representation. Let us first consider the simple case of $\ket{\Omega_S}$:
\begin{theorem}
Let  $\mathcal{W}$ be the Weyl $C^*$-algebra for GFT. For the squeezed vacuum state $\ket{\Omega_S}$ on $\mathsf{W}=\mathcal{W}\otimes{\mathcal{W}}'$, there exist a Bogoliubov transformation such that $\ket{\Omega_S}$ is transformed into the form of a bosonic TFD state with respect to the KMS states $\omega_{\text{KMS}}$ on $\mathcal{W}$:
\begin{equation}\label{42}
\ket{\Omega_S}=\frac{1}{\sqrt{Z}}\sum_{n=0}^\infty e^{-in\beta u_{*}/2}\ket{\{g_i\},n}\ket{\{g'_i\},\tilde{n}}
\end{equation}
where $u_{*}$ is an eigenvalue of $U_*$ \eqref{u}, and the Fock states are the GNS Fock states.
\end{theorem}
\begin{proof}
In the simpler case where $b\in\mathbb{R}$,  the thermal vacuum  \eqref{40} has the same form as in \eqref{37}, and hence the Bogoliubov transformation similar to \eqref{bogo} gives the TFD states in the standard form. 

In the more general algebraic case, we need to find the *-automorphism $\alpha_T$ on $\mathsf{W}$ such that $\alpha_T(W(\psi))=W(T\psi)$ for $W(\psi)\in{\mathsf{W}}$ and a  symplectic transformation $T$ on the space of  $\psi$'s. 
Notice that in \eqref{40} $b<\infty$, so by Theorem 2 of \cite{HR97} there exists a unitary $U$ on the GNS  space $\mathcal{H}_\mathsf{W}$ such that  $\pi_\mathsf{W}(\alpha_T(W))=U\pi_\mathsf{W}(W)U^*$. 
Denote $\mathfrak{G}\equiv\frac{b}{2}\Psi^\dag(\psi)\tilde{\Psi}^\dag(\psi')+\frac{\bar{b}}{2}\Psi(\psi)\tilde{\Psi}(\psi')$ and  consider $\pi_\mathsf{W}=\pi_{\mathsf{W},F}$ as the Fock representation, then we have  $U=e^{-it\mathfrak{G}}$.
Therefore, we can choose
\[\alpha_T(W(\psi))=UW(\psi)U^*=e^{-it\mathfrak{G}}W(\psi)e^{it\mathfrak{G}}=W(T\psi),\]
with the symplectic transformation 
\[
T=\cosh(\wp)+\jmath\sinh(\wp'),\quad \wp= \delta\bigl(\psi(\{g_i\}),\abs{b}\bigr),~ \wp'= \delta\bigl(\psi'(\{g_i\}),\abs{b}\bigr)
\]
where  $\jmath=-1$ for the bosonic case.
When acting on the field operators, $\alpha_T$ gives the Bogoliubov transformations of the field operators on the Fock space $\mathcal{H}_{\mathsf{W},F}$,
\begin{equation}\label{pprime}
\alpha_T(\Psi(\psi))=\Psi(\cosh\wp\psi)+\tilde{\Psi}^{\dag}(-\sinh\wp'\psi')=\phi\cosh\abs{b}-\tilde{\phi}^\dag\sinh\abs{b}.
\end{equation}
By choosing $b$ as in \eqref{bogo} and using the Gibbs states $\rho_g$  on $\mathcal{W}$, we can obtain the TFD states.
\begin{equation}\label{stan}
\ket{\Omega_S}=\sqrt{1-e^{-i\beta u_*}}\sum_{n=0}^\infty e^{-in\beta u_*/2}\ket{\{g_i\},n}\ket{\{g'_i\},\tilde{n}}.
\end{equation}
This is because by  setting the Bogoliubov operator $\chi$ as $\cosh2\chi=\coth({i\beta u_*}/2)$, we can rewrite the KMS state $\omega_{\text{KMS}}$ on $\mathcal{W}$ in the Fock representation as $\omega_{\text{KMS}}(W(\psi))=e^{-\cosh2\chi(\psi,\psi)/4}$. Since $\cosh2\chi=\cosh^2\chi+\sinh^2\chi$, we have
\[\omega_{\text{KMS}}(W(\psi))=e^{-\frac{1}{4}(\cosh^2\chi+\sinh^2\chi)(\psi,\psi)}=\omega_F(W(\cosh\chi\psi))\otimes\omega_F(W(\sinh\chi\psi)).
\]
Hence, by choosing $\chi\psi=\abs{b}\psi$, we can obtain the TFD form \eqref{stan}.
Since in equilibrium the Gibbs state $\rho_g$ is invariant under the group-translation flow $\alpha_{\{g_i\}}$, it is safe for the purpose of proving existence to take a constant eigenvalue $u_*$ of $U_*$ in $\rho_g$. In this way, the partition function of $\rho_g$ becomes $Z=\sum_ne^{-\beta g_n}=(1-e^{-i\beta u_*})^{-1}$, and $\ket{\Omega_S}$ takes the standard form. 
\end{proof}
\begin{corollary}
Let $\ket{\Omega_S}$ be as above. Then both $\omega_S$ on $\mathsf{W}$ and the reduced state to $\mathcal{W}$ are KMS with respect to a  group-translation flow $\alpha_{\{g_i\}}$ on the group fields.
\end{corollary}
\begin{proof}
Let $\rho_S=\ket{\Omega_S}\bra{\Omega_S}$. By using $\braket{\{g'_i\},\tilde{n}|\{g'_i\},\tilde{m}}=\delta_{nm}$, we retain the Gibbs states on $\mathcal{W}$,
\[
\rho_\mathcal{W}=\text{Tr}_{\mathcal{W}'}\rho_S=\frac{1}{Z}\sum_ne^{-in\beta u_*}\ket{\{g_i\},n}\bra{\{g_i\},n}.\]
Consider the group-translation flow $\alpha_{\{g_i\}}$ on $\mathcal{W}$ with the same generator $\mathcal{G}=iU_*$ as in $\ket{\Omega_S}$. The KMS condition is satisfied for Gibbs states following the standard arguments \cite{Das97}, with the Hamiltonian replaced by the generator $\mathcal{G}$.
\end{proof}

 We remark that the single-mode expression \eqref{42} can be changed to the multi-mode expression by considering the multi-mode expansion of the test functions $\psi$. For instance, consider the expansion of $\psi(g)$ with respect to the SU(2) group basis, then
\begin{equation}
\Psi(\psi)=\int_{G}\prod_{i=1}^nd\mu(g) \overline{\sum_j\text{Tr}[\psi^jD^j(g)])} \phi(\{g_i\}),\quad j\in\frac{\mathbb{N}}{2},
\end{equation}
where $D^j$ is the Wigner matrix. Since in \eqref{pprime} we have identified $\abs{b}$ with $\psi$, now the squeeze parameter $b$ also has an expansion into $b_j$, and as a consequence, the eigenvalue $u_*$ has a corresponding expansion into $u_{*,j}$. Then
\begin{equation}
\ket{\Omega_S}\sim\sum_{n=0}^\infty\sum_{j\in\mathbb{N}/2} e^{-in\beta u_{*,j}/2}\ket{\{g_i\},n}\ket{\{g'_i\},\tilde{n}}.
\end{equation}
Moreover, since the CCR algebra structure is the same for all finite or infinite $N$, the $\ket{\Omega_S}$ can be extended to the full Fock space. Written in terms of the multi-particle occupation number basis, it becomes
\begin{equation}
\ket{\Omega_S}=\sum_{N_i=1}^\infty\sum_{n_i=0}^{N_i}\sum_{j\in\mathbb{N}/2} M(N_i,n_i,j) e^{-in_i\beta u_{*,j}/2}\ket{\{g_i\};\{n_i\}}\ket{\{g'_i\};\{\tilde{n}_i\}}.
\end{equation}
where $M(N_i,n_i,j)$ is a normalization factor and $N_i$ is the number of particles in the subspace $\mathcal{H}^{N_i}$.

The above TFD states are constructed as the squeezed vacua on some known product algebra $\mathsf{W}$ which does not obviously correspond to the doubled Hilbert-space formulation of TFD (since $\mathcal{M}\neq\mathcal{M}'$).
 In view of the  localized finite-dimensional version of TFD (cf. Theorem \ref{req}), we consider alternatively  the TFD for local normal folium: Consider the GFT algebras $\mathcal{B}(\mathcal{H}_{N_1}),\mathcal{B}(\mathcal{H}_{N_2})$ with $N_2>N_1$ and $\mathcal{B}(\mathcal{H}_{N_1})\subset\mathcal{B}(\mathcal{H}_{N_2})$. Suppose the split property holds, then there is an intermediate type I factor $\mathcal{B}$ such that  $\mathcal{B}(\mathcal{H}_{N_1})\subset\mathcal{B}\subset\mathcal{B}(\mathcal{H}_{N_2})$ and $\overline{\mathcal{B}\lor\mathcal{B}'}=\mathcal{B}(\mathcal{H})$. The states in localized space $\mathcal{H}_\mathcal{B}$ are in fact localized in $\mathcal{B}(\mathcal{H}_{N_2})$: for $\psi\in\mathcal{H}_\mathcal{B}$ and $B'\in\mathcal{B}({\mathcal{H}_{N_2}})'$,
 \[(\psi,B'\psi)=(W^*(\psi\otimes\Omega),B'W^*(\psi\otimes\Omega))=
 ((\psi\otimes\Omega),WB'W^*(\psi\otimes\Omega))=(\Omega,B'\Omega).
 \]
 And similarly, the states in   $\mathcal{H}_{\mathcal{B}'}$ are  localized in $\mathcal{B}(\mathcal{H}_{N_1})'$. As the operators in $\mathcal{B}$ and $\mathcal{B}'$ are related by $J$, the operator contents on $\mathcal{H}_{\mathcal{B}}$ and $\mathcal{H}_{\mathcal{B}'}$ are approximately the same up to the ``boundary'' operators.
 Therefore, by Theorem \ref{req} we have the finite-dimensional TFD states in  $\mathcal{H}_{\mathcal{B}}\otimes\mathcal{H}_{\mathcal{B}'}$ conforming to the doubled Hilbert-space formulation.
 
Notice that in  equilibrium the ``temperature'' parameter $\beta$'s  in the global and the localized cases should be the same. But in more general nonequilibrium cases
 the KMS condition for the states localized  $\mathcal{H}_{\mathcal{B}}\otimes\mathcal{H}_{\mathcal{B}'}$ could be  different from the global KMS condition in more general nonequilibrium cases, e.g. with different  $\beta$'s. 
 
\subsection{Graphic representation}
Given a gauge-invariant group field $\phi(\{g_i\})^\dag$, we can interpret it  as creating a dual quantum polyhedron in the quantum geometry of LQG with each $g_i$ on the edge dual to a face of the polyhedron. In the second quantization formalism, this gives rise to a single-polyhedron Fock state (in terms of field operators) $\Psi(\{g_i\})^\dag\ket{\Omega}\equiv\ket{\triangle^1}$  and likewise for $\ket{\triangle^n}$ of $n$ independent polyhedra. These independent polyhedra can be glued 
according to the interaction term \eqref{9} in the GFT action or  the tracial observables \eqref{15}. 
These group fields constitute the Weyl algebra $\mathcal{W}$ or the von Neumann algebra $\mathcal{M}$. 
Now the TFD extension adds $\tilde{\mathcal{M}}=\mathcal{M}'$ to form $\mathcal{M}\otimes\mathcal{M}'\equiv\mathsf{M}$. For reasons explained in Sec. \ref{s3.1}, $\mathcal{M}'$ differs from $\mathcal{M}$ in general but with the same Fock space structure, which means the  field operators  $\tilde{\Psi}(\{g'_i\})$ for $\mathcal{M}'$ and $\Psi(\{g_i\})$ (or the wavefunctions $\psi$) are different functions on $G^{\times n}$, and they satisfy $[\Psi(\{g_i\}),\tilde{\Psi}(\{g'_i\})]=0$ and  $[\Psi(\{g_i\}),\tilde{\Psi}(\{g'_i\})^\dag]=0$.\footnote{In  the conventional formuation of second quantization, the field operators are independent of the choice of the complete set of single-body wavefunctions, if two different sets of creation/annihilation operators can be related, e.g. by $b=\sum_i\psi_j^*\psi_ia_i$. Here we do not have such transformations as the original and ``tilde" field operators are commutative to each other.
} 
 The structure of $\mathsf{M}$ allows us to take $\tilde{\phi}=\phi$, so that the differences between $\Psi$ and $\tilde{\Psi}$ are solely in the test wavefunctions $\psi\neq{\psi'}$.
 The Fock states on $\mathsf{M}$ are then doubled to be $\ket{\blacktriangle^n}$, e.g. 
\[\Psi(\{g_i\})^\dag\otimes\tilde{\Psi}(\{g'_i\})^\dag\ket{\Omega}\otimes\ket{\Omega}\equiv\ket{\blacktriangle^1}.\]
The doubled states $\ket{\blacktriangle^n}$ are further transformed into the entangled squeezed vacuum state 
\begin{equation}\label{43sup}
\ket{\Omega_S}\equiv\ket{\blacktriangle}_{\text{TFD}}=\sum_ne^{-\beta\mathcal{G}_n}\ket{\blacktriangle^n},\end{equation}
 which is the TFD state for  a GFT in equilibrium with respect to a flow on group fields. Note that in the proof of \eqref{42}  the differences in $\psi,\psi'$ have been wiped out in the Bogoliubov transformation \eqref{pprime}, whereby the differences in the observable algebras are invisible at the state level. 

By  definition, $\mathcal{M}=\pi_\omega^{\prime\prime}(\mathcal{W})\subset \mathcal{B}(\mathcal{H}_\omega)$ are constructed from the representation of the  Weyl algebra of the gauge-invariant group fields, and hence both $\mathcal{M}$ and $\mathcal{M}'$ contain gauge-invariant operators in $\mathcal{B}(\mathcal{H}_\omega)$. 
Consequently,  we can envision each quantum polyhedron originally created by $\Psi(\{g_i\})^\dag$ to be decorated by an additional layer of polyhedral data created by  $\tilde{\Psi}(\{g_i'\})^\dagger$, and then they are  superposed into a TFD state.
As an example, consider the a single-tetrahedron state $\ket{\triangle^1}$ created by $\phi(\{g_i\})^\dag,i=1,2,3,4$. Then the Fock states on $\mathsf{M}$ can be depicted (in  the dual quantum geometry) as the ``fat" graphs: 
\begin{equation}\label{316}
\ket{\triangle^1}= \ket{\vcenter{
\xymatrix@-0.9pc{
&  \ar@{-}[d] \\
& *-<2pt>[o][F-]{ }\ar@{-}[d]\ar@{-}[r]\ar@{-}[l] &\\
&&}}}\quad\Rightarrow\quad\ket{\blacktriangle^1}=
\ket{\vcenter{
\xymatrix@-1pc{
&  \ar@{=}[d] \\
& *-<2pt>[o][F-]{ \bullet}\ar@{=}[d]\ar@{=}[r]\ar@{=}[l]&\\
&&}}}.
\end{equation}
And the TFD states are obtained by the superposition $\ket{\blacktriangle}_{\text{TFD}}$ \eqref{43sup}.

A rough reason for still putting the  states of $\mathcal{M}'$ on the original quantum geometry is that the $\mathcal{M}'$ is obtained from $\mathcal{M}$ and they are both subsets of $\mathcal{B}(\mathcal{H}_\omega)$. For simplicity, we can identify $g_i=g'_i$, so that they live on the same edge. In the homogeneous GFT states discussed in the next section, this identification is given by the wavefunction homogeneity condition.
 The superposition in $\ket{\blacktriangle}_{\text{TFD}}$ indicates another reason for this:
a single-polyhedron ``tilde" state $\tilde{\Psi}(\{g'_i\})^\dag\ket{\Omega}$ is expected to correspond to a single-polyhedron $\Psi(\{g_i\})^\dag\ket{\Omega}$, in the sense of particle-hole excitations in a ``Dirac sea''.  The superposition of these particle-hole excitations constitutes the collective TFD vacuum, which could be interpreted as the superposed states of vacuum fluctuations \cite{Req13}. In this sense, the fat graphs  can be intuitively understood as the quantum polyhedra  with vacuum fluctuations taken into consideration.

\section{Group field theory condensates at finite $\beta$}\label{s4}
The GFT condensates are proposed to approximate the continuum geometries as the BEC of the group fields that encode the fundamental building blocks of simplicial quantum gravity. As typical collective excitations, the squeezed states created by group fields in the second quantization formalism also rise to the  GFT condensate states. Since the TFD states are two-mode squeezed states on the Fock vacuum, they are also GFT condensate states.
We will show that the equilibrium parameter $\beta$ distinguishes such TFD condensates from other types of GFT condensates.

We start with the  homogeneous GFT condensate states in the usual sense \cite{GOS13}:
a GFT condensate state is said to be homogeneous if the 
$n$-``polyhedra'' GFT states in its  series expansion are product  states  with the same wavefunction $\psi$ assigned to each single-``polyhedron'' GFT state in the product. 
A homogeneous GFT condensate state then corresponds to a homogeneous continuum (quantum) space characterized by the single collective wavefunction $\psi$.
The homogeneous GFT condensate states in this sense do not contain additional topological information such as the gluings of graphs. To include the topological information encoded in the graphs, the generalized homogeneous GFT condensate states  are defined in  \cite{OPRS15}. In \cite{OPRS15}, the  homogeneity condition is  for the vertex wavefunctions, that is, the single-vertex/polyheron  states (in the  series expansion of a generalized homogeneous GFT condensate state)   are created by the GFT field operators having the same single-``polyhedron'' wavefunction $\psi$; the gluing information of graphs is encoded in the convolutions of the auxillary group elements $h$ in the vertex wavefunctions $\psi(hg)$. (Left and right gauge invariance should be imposed on $\psi$ to constrains the contribution from the additional information.) In this way,   we have in effect the homogeneous wavefunctions over discrete simplicial geometries, which allows  the coarse-grainings or fine-grainings  by changing the graphs, while retaining the vertex wavefunction homogeneity.
In this section, by homogeneous we mean the latter vertex wavefunction homogeneity allowing gluings.

Now for TFD states, we have the similarly  ``tilde'' homogeneous GFT condensate states; 
we also expect that there will be a condition for matching the ``tilde'' and ``untilde'' parts of the TFD states, so that the TFD states can have the graphic representation as \eqref{316}.
To this end, we consider the  matching condition on the ``tilde'' and ``untilde'' group field operators: $\Psi(\{g_i\})$ and $\tilde{\Psi}(\{g_i\})$ are matched if their vertex wavefunctions are the same, i.e. $\psi=\tilde{\psi}$. This way, the two distributions of group data over the polyhedra they generate are the same, so that we can represent them as doubled fat graphs as \eqref{316}. On the other hand, a TFD state for GFT can be formally considered as 
 a squeezed state on the  doubled Fock vacuum $\ket{\Omega}\otimes\ket{\Omega}$. However, by matching  $\Psi(\{g_i\})$ and $\tilde{\Psi}(\{g_i\})$  and putting these two parts on the same graph, we  have in effect a ``single-pole'' condensate on   the single-vacuum $\ket{\Omega}$. In other words, the discrete geometry at finite $\beta$ is still given by the GFT states on the single-vacuum $\ket{\Omega}$, but now decorated by the new data specified by the flow. 
 In a similar way, we can  consider for example  the dipole GFT condensate states on a single-vacuum, and represent the corresponding   discrete geometry at finite $\beta$ by the fat graphs.

The  TFD states add twofold new data about the flow to the  GFT condensates: one is the ``tilde"   states   which can been put on the original quantum geometry; the other is the Gibbs states determining the squeezing induced by the flow. As a consequence, in addition to the homogeneity considerations, we also have the equilibrium condition:
\begin{definition}
In a GFT with field operators $\Psi(\{g_i\}),\Psi^\dag(\{g_i\}),i\geqslant4$ and a  group-translation flow $\alpha_{\{g_i\}}$, the GFT states are said to be in equilibrium with respect to $\alpha_{\{g_i\}}$ at the equilibrium parameter $\beta$ if their TFD states have the same parameter $\beta$.
\end{definition}
At a uniform $\beta$, one can choose a homogeneous seed state $\ket{\text{seed}}$  of GFT condensate and then  refine it topologically,  e.g. $\hat{r}\ket{\text{seed}}$ as in \cite{OPRS15}.   
Let us  consider instead the TFD states with different $\beta$.  In doing this, we will be in a globally nonequilibrium situation where the local equilibrium is retained.\footnote{In a similar situation in algebraic quantum field theories, these $\beta$'s are not only equilibrium parameters but also order parameters for symmetry breaking. Cf. \cite{Oji03}.}

Since the the parameters $\beta$ or the Gibbs states are new data added to the quantum geometry represented by GFT, we can specifically consider different condensate states without these new data related by the refinement operators $\hat{r}$, and then add to each layer $r$ of refinement the Gibbs states with a fixed $\beta_r$. Consequently, on each layer of some shell condensate states there is a distinct algebraic symmetry with equilibrium parameter $\beta$.\footnote{Note that when gluing different condensate states, the gluing edge generically breaks the closure constraint (closure defect) \cite{AC16}. The problem of how to relate different algebraic symmetries defined with different $\beta$ does not arise in that the presence of closure defects breaks the gauge invariance of group fields, and hence the flow defined on the gauge-invariant  group fields cannot be defined across the boundary.}
 If $\beta$ can be interpreted as the inverse temperature, we then envisage a  relation between  $\beta$ and the refinement layer $r$ in analogy to the Tolman-Ehrenfest effect for equilibrium temperatures in a gravitational field.

Without loss of generality, let us work in the basic case of 4-fold group fields $\phi(\{g_i\}),i=1,2,3,4,$ representing the 4-valent vertices dual to quantum tetrahedra. The topological connectivity between group fields in a TFD condensate state is still described the edge-homogeneity  constraints.
 So the white-black colored group fields  are now decorated by the ``tilde"  states:
\begin{equation}
\ket{\vcenter{
\xymatrix@-0.9pc{
&  \ar@{-}[d] \\
& *-<1pt>[o][F-]{ \circ}\ar@{-}[d]\ar@{-}[r]\ar@{-}[l] &\\
&&}}}\equiv
\ket{\vcenter{
\xymatrix@-1pc{
&  \ar@{=}[d]^{1} \\
& *-<1pt>[o][F-]{ \bullet}\ar@{=}[d]_{3}\ar@{=}[r]_{4}\ar@{=}[l]_{2} &\\
&&}}},\quad
 \ket{\vcenter{
\xymatrix@-0.9pc{
&  \ar@{-}[d] \\
& *-<2pt>[o][F-]{ \bullet}\ar@{-}[d]\ar@{-}[r]\ar@{-}[l] &\\
&&}}}\equiv
\ket{\vcenter{
\xymatrix@-1pc{
&  \ar@{=}[d]^{1} \\
& *-<2pt>[o][F-]{ \bullet}\ar@{=}[d]_{3}\ar@{=}[r]_{2}\ar@{=}[l]_{4} &\\
&&}}}.
\end{equation}
Since the ``tilde" states are constrained by the edge-matching constraint,
 one can in principle refine the geometry by adding such colored vertices  in the same way as \cite{OPRS15}. The homogeneity condition  constrains  the  connectivity of group fields (or their dual quantum polyhedra) in a GFT condensate, but not the parameter $\beta$ in a GFT condensate state of  TFD form. For two group fields homogeneously connected in a TFD state at an equilibrium parameter $\beta$, we can depict the TFD superposition of two GFT Fock states by an additional dotted line, e.g. 
 \begin{equation}
\ket{\vcenter{
\xymatrix@-0.9pc{
&  \ar@{-}[d] \\
&  *-<1pt>[o][F-]{ \circ}\ar@{-}[d]\ar@{-}[r]\ar@{-}[l] &\\
&&}}}+
 \ket{\vcenter{
\xymatrix@-0.9pc{
&  \ar@{-}[d] \\
& *-<2pt>[o][F-]{ \bullet}\ar@{-}[d]\ar@{-}[r]\ar@{-}[l] &\\
&&}}}\quad\Rightarrow\quad
\ket{\vcenter{
\xymatrix@-0.9pc{
&  \ar@{-}[d]&\ar@{-}[d] \\
& *-<1pt>[o][F-]{ \circ}\ar@{-}[d]\ar@{-}[r]\ar@{-}[l] &*-<2pt>[o][F-]{ \bullet}\ar@{-}[d]\ar@{-}[r]\ar@/^/@{.}[l]^{\beta}&\\
&&&}}}
\end{equation}
where the gluing rules are obeyed such that only the edges on vertices of different colors (white$\neq$black)  with the same referencial index (e.g. $4=4$)  can be glued. In this way, the equilibrium condition constrains the algebraic symmetry of the group fields on top of the homogeneous condensate geometries and their refinements. 

Now the TFD states with the same GFT condensate structure can be distinguished by the different $\beta$s  parameterizing different algebraic symmetries. If we consider different local realizations of the TFD states in the finite-dimensional products of two Hilbert spaces, it is possible to find different local $\beta$-parameters  within a single global algebra. In the following, we consider examples with the different layers of GFT shell condensates and dress them with different local $\beta$s.
\begin{example}\label{e4.4}
Consider the homogeneous shell condensate states in GFT, which can model the horizon of a black hole \cite{OPS16}. The seed state for a homogeneous shell condensate is chosen to be the following state
 \begin{equation}
\ket{\text{seed}}_b=
\ket{\vcenter{
\xymatrix@-0.8pc{
& *-<1pt>[o][F-]{ \circ}\ar@/^/@{-}[d]\ar@/_/@{-}[d]\ar@{-}[r]\ar@{-}[l]_4 &*-<2pt>[o][F-]{ \bullet}\ar@/^/@{-}[d]\ar@/_/@{-}[d]\ar@{-}[r]&*-<1pt>[o][F-]{ \circ}\ar@{-}[r]_{1}\ar@/^/@{-}[d]\ar@/_/@{-}[d]&\\
&*-<2pt>[o][F-]{ \bullet}\ar@{-}[r]\ar@{-}[l]_4&*-<1pt>[o][F-]{ \circ}\ar@{-}[r]&*-<2pt>[o][F-]{ \bullet}\ar@{-}[r]_1&}}},
\end{equation}
with two boundaries, i.e. two sets of open edges with indices 1 and 4 respectively. Since the boundaries consist of open edges, we do not consider any flow on the boundary edges. 
 The refinement operators are 
\begin{equation}\label{4.5}
\hat{r}_W:\ket{\vcenter{
\xymatrix@-0.9pc{
&  \ar@{-}[d] \\
& *-<1pt>[o][F-]{ \circ}\ar@{-}[d]\ar@{-}[r]\ar@{-}[l] &\\
&&}}}\mapsto
\ket{\vcenter{
\xymatrix@-0.9pc{
&  \ar@{-}[d]&& \\
& *-<1pt>[o][F-]{ \circ}\ar@{-}[d]\ar@{-}[r]\ar@{-}[l] &*-<2pt>[o][F-]{ \bullet}\ar@{-}[r]&*-<1pt>[o][F-]{ \circ}\ar@{-}[r]\ar@/^/@{-}[l]\ar@/_/@{-}[l]&\\
&&&&}}},\quad
\hat{r}_B: \ket{\vcenter{
\xymatrix@-0.9pc{
&  \ar@{-}[d] \\
& *-<2pt>[o][F-]{ \bullet}\ar@{-}[d]\ar@{-}[r]\ar@{-}[l] &\\
&&}}}\mapsto
\ket{\vcenter{
\xymatrix@-0.9pc{
&  \ar@{-}[d]&& \\
&*-<2pt>[o][F-]{ \bullet}\ar@{-}[d]\ar@{-}[r]\ar@{-}[l] &*-<1pt>[o][F-]{ \circ}\ar@{-}[r]&*-<2pt>[o][F-]{ \bullet}\ar@{-}[r]\ar@/^/@{-}[l]\ar@/_/@{-}[l]&\\
&&&&}}},
\end{equation}
and at the boundaries, for example a boundary vertex with index 4,
\begin{equation}
\hat{r}_{W,\partial}:\ket{\vcenter{
\xymatrix@-0.9pc{
&  \ar@{-}[d]^4 \\
& *-<1pt>[o][F-]{ \circ}\ar@{-}[d]\ar@{-}[r]\ar@{-}[l] &\\
&&}}}\mapsto
\ket{\vcenter{
\xymatrix@-0.9pc{
&  \ar@{-}[d]^4&\ar@{-}[d]^{4'}& \ar@{-}[d]^{4^{\prime\prime}}\\
& *-<1pt>[o][F-]{ \circ}\ar@{-}[d]\ar@{-}[r]\ar@{-}[l] &*-<2pt>[o][F-]{ \bullet}&*-<1pt>[o][F-]{ \circ}\ar@{-}[r]\ar@/^/@{-}[l]\ar@/_/@{-}[l]&\\
&&&&}}},\quad
\hat{r}_{B,\partial}: \ket{\vcenter{
\xymatrix@-0.9pc{
&  \ar@{-}[d]^4 \\
& *-<2pt>[o][F-]{ \bullet}\ar@{-}[d]\ar@{-}[r]\ar@{-}[l] &\\
&&}}}\mapsto
\ket{\vcenter{
\xymatrix@-0.9pc{
&  \ar@{-}[d]^4&\ar@{-}[d]^{4'}& \ar@{-}[d]^{4^{\prime\prime}}\\
&*-<2pt>[o][F-]{ \bullet}\ar@{-}[d]\ar@{-}[r]\ar@{-}[l] &*-<1pt>[o][F-]{ \circ}&*-<2pt>[o][F-]{ \bullet}\ar@{-}[r]\ar@/^/@{-}[l]\ar@/_/@{-}[l]&\\
&&&&}}},
\end{equation}
where the open edges on the boundary are also refined. The refined shell condensate states are formally 
\begin{equation}
\ket{r}=\prod_rf_r(\hat{r}_{B},\hat{r}_{W},\hat{r}_{B,\partial},\hat{r}_{W,\partial})\ket{\text{seed}}_b
\end{equation}
 where $f_r$ is a function of the refinement operators.

 Given an infinite number of possible refinements of a GFT condensate, one can choose a particular refinement level $r_h$ of the shell condensate  to model the horizon of a (quantum) black hole. In  a generic shell condensate $\ket{r}$, the boundary with index 4 is not only the outer boundary of  $\ket{r}$ but also the  inner boundary of  $\ket{r+1}$ at the next level.  But the connectivity between two shell condensate does not specify a black hole horizon. We notice that, with hindsight, the event horizon of a black hole possesses the surface gravity $\kappa$ that relates the symmetry of the horizon to the temperature of Hawking radiation, and in particular, the {\it local} surface gravity $\bar{\kappa}=1/\ell$, where $\ell$ is the proper distance from an observer (or an apparent horizon) to the horizon \cite{FGP11}.
 Therefore, for the shell condensate at a chosen refinement level to be able to model a black hole horizon, we  expect the equilibrium parameter $\beta$ to play the role of $\bar{\kappa}$, encoding the symmetry of the horizon as well as the local Hawking temperature.

For instance, consider the horizon shell $\ket{r_h}$ at equilibrium parameter $\beta_h$, and its next refinement $\ket{r_h+1}=f_{r_h+1}(\hat{r}_{B},\hat{r}_{W},\hat{r}_{B,\partial},\hat{r}_{W,\partial})\ket{r_h}$ at $\beta'_h$. Parts of the boundary looks like the following
\begin{equation}\label{52}
\ket{\vcenter{
\xymatrix@-0.9pc{
&  \ar@{-}[d]&\ar@{-}[d] \\
& *-<1pt>[o][F-]{ \circ}\ar@{-}[d]\ar@{-}[r]\ar@{-}[l] &*-<2pt>[o][F-]{ \bullet}\ar@{-}[d]\ar@{-}[r]\ar@/^/@{.}[l]^{\beta}&\\
&&&}}}+
\ket{\vcenter{
\xymatrix@-0.9pc{
&  \ar@{-}[d]&\ar@{-}[d]& \ar@{-}[d]& \ar@{-}[d]\\
& *-<2pt>[o][F-]{ \bullet}\ar@{-}[d]\ar@{-}[r]\ar@{-}[l] & *-<1pt>[o][F-]{ \circ}\ar@/^/@{.}[l]^{\beta'}&*-<2pt>[o][F-]{ \bullet}\ar@{-}[r]\ar@{-}[l]\ar@/_/@{-}[l]\ar@/^/@{.}[l]^{\beta'}&*-<1pt>[o][F-]{ \circ}\ar@{-}[r]\ar@{-}[d]\ar@/^/@{.}[l]^{\beta'}&\\
&&&&&}}}\quad\Rightarrow\quad
\vcenter{
\xymatrix@-0.9pc{
&  \ar@{-}[d]&\ar@{-}[d]& \ar@{-}[d]& \ar@{-}[d]\\
&*-<2pt>[o][F-]{ \bullet}\ar@{-}[d]_{g_l}\ar@{-}[r]\ar@{-}[l] &*-<1pt>[o][F-]{ \circ}\ar@/^/@{.}[l]^{\beta'}&*-<2pt>[o][F-]{ \bullet}\ar@{-}[r]\ar@{-}[l]\ar@/_/@{-}[l]\ar@/^/@{.}[l]^{\beta'}&*-<1pt>[o][F-]{ \circ}\ar@{-}[r]\ar@{-}[d]\ar@/^/@{.}[l]^{\beta'}&(r_{h}+1)\\
& *-<1pt>[o][F-]{ \circ}\ar@{-}[d]\ar@{-}[l]\ar@{-}[rrr]&&&*-<2pt>[o][F-]{ \bullet}\ar@{-}[r]\ar@{-}[d]\ar@/^/@{.}[lll]^{\beta}&(r_h)\\
&&&&}}
\end{equation}
In \eqref{52} we require $\beta'>\beta$, or equivalently $\beta$ is a monotonically increasing function $\beta(r)$ of  the refinement level $r$. Then the $\beta'$ can correspond to the inverse of the local surface gravity $\bar{\kappa}$, if  the proper distance $\ell$ is also a monotonically increasing function $\ell(r)$ of the refinement level $r$. In other words, we can take 
\begin{equation}\beta'(r_h+1)\propto\ell(r_h+1),\end{equation}
 so that $\beta'$ outside the horizon shell can be interpreted as the inverse local Hawking temperature. 
Meanwhile, in \eqref{52} the edge labeled by the group element $g_l$ is dual to a quantum area patch in the LQG quantum geometry, which gives rise to the constraint that the sum over the area eigenvalues in $\ket{r_h}$ is the area $A$ of the horizon.
 The proper distance $\ell$ should be defined  with respect to this horizon, but this is not given in the quantum geometries. The TFD states of GFT, {\it a posteriori}, define the $\ell$ in terms of the equilibrium parameter $\beta$.
\end{example}
We remark that the use of TFD states in  Example \ref{e4.4} is different from \cite{Isr76}: on the one hand, here a TFD  state only defines single GFT shell condensate decorated by  ``tilde" states, instead of a double states; on the other hand, the TFD state entangles the vacuum fluctuation modes or the ``particle-hole'' hole excitations, which is not necessarily the trans-horizon entanglement.
\begin{example}
Let us turn to the GFT condensate quantum cosmology. For a GFT coupled with massless scalar fields $\varphi$, the group fields are defined on $G^{\times n}\times\mathbb{R}^{n}$. The scalar fields can be used as a relational clock if they are monotonic on $\mathbb{R}^n$, and then the GFT can be deparameterized to have a relational Hamiltonian generating the dynamics \cite{WE19}. This generic method, when applied to the cosmological case, has been illustrated in a toy model \cite{AGW18} where the deparameterized Hamiltonian is the squeezing operator (in the spin basis). 
In GFT at equilibrium parameter $\beta$,  the TFD state $\ket{\Omega_S}$ can be written in the form of a squeezed  state as \eqref{40},
\begin{equation}\label{4.10}
\ket{\Omega_S}=e^{i\varphi\{\frac{b}{2}\Psi^\dag(\psi)\tilde{\Psi}^\dag(\psi')+\frac{\bar{b}}{2}\Psi(\psi)\tilde{\Psi}(\psi')\}}\ket{\Omega}\otimes\ket{\Omega},\quad \varphi\in\mathbb{R},b\in\mathbb{C},
\end{equation}
where we have written the parameter explicitly as the scalar field $\varphi$. By identifying the relational  Hamiltonian (in the second-quantized form) $H_\varphi=-\frac{b}{2}\Psi^\dag(\psi)\tilde{\Psi}^\dag(\psi')-\frac{\bar{b}}{2}\Psi(\psi)\tilde{\Psi}(\psi')$ with respect to $\varphi$, we see that the TFD squeezed state is the ``time" evolution of the Fock vacuum generated by $H_\varphi$. In other words, $\ket{\Omega_S}$ is a solution of the relational Schr\"odinger equation $i\frac{d}{d\varphi}\ket{\Omega_S}=H_\varphi\ket{\Omega_S}$. 

Notice that the scalar fields are not introduced in the  setup of this paper, the $\varphi$ in \eqref{4.10} are in fact functions $\varphi(\beta)$ of the equilibrium parameter $\beta$. But we can   consider the flow parameter $\beta(\varphi,b,n)$ as a monotonic function of $\varphi$ such that $\varphi(\beta)$ is also monotonic as an inverse function.

The relational evolution generated by $H_\varphi$ should be applicable to arbitrary initial states. Let us take  the seed state for spatial 3-sphere as the initial state,
 \begin{equation}
\ket{\text{seed}}_c=
\ket{\vcenter{
\xymatrix@-0.8pc{
& *-<1pt>[o][F-]{ \circ}\ar@/^/@{-}[d]\ar@/_/@{-}[d]\ar@{-}[r]\ar@{-}[d] &*-<2pt>[o][F-]{ \bullet}\ar@/^/@{-}[d]\ar@/_/@{-}[d]\ar@{-}[d]&\\
&*-<2pt>[o][F-]{ \bullet}\ar@{-}[r]&*-<1pt>[o][F-]{ \circ}&}}}.
\end{equation}
Then the squeezing   $e^{i\varphi H_\varphi}\ket{\text{seed}}_c\equiv\ket{\omega_S}$ generates a series of new vertices, where each term $\Psi^\dag(\psi)\tilde{\Psi}^\dag(\psi')$ creates a colored vertex. So, in contrast to the model in \cite{AGW18}, here the squeezing only refines $\ket{\text{seed}}_c$ in the same way as \eqref{4.5}. Again, we impose the additional constraint that at each  refinement level $r$ the GFT condensate state should be in equilibrium at $\beta_r$, which implies a equal-$\varphi(\beta_r)$ spatial hypersurface. Therefore, the refinement maps changes the equilibrium parameter $\beta$, e.g.
\begin{equation}
\ket{\vcenter{
\xymatrix@-0.9pc{
&  \ar@{-}[d]&\ar@{-}[d] \\
& *-<1pt>[o][F-]{ \circ}\ar@{-}[d]\ar@{-}[r]\ar@{-}[l] &*-<2pt>[o][F-]{ \bullet}\ar@{-}[d]\ar@{-}[r]\ar@/^/@{.}[l]^{\beta}&\\
&&&}}}\quad\mapsto\quad
\ket{\vcenter{
\xymatrix@-0.9pc{
&  \ar@{-}[d]&& & \ar@{-}[d]\\
&  *-<1pt>[o][F-]{ \circ}\ar@{-}[d]\ar@{-}[r]\ar@{-}[l] &*-<2pt>[o][F-]{ \bullet}\ar@/^/@{.}[l]^{\beta'}& *-<1pt>[o][F-]{ \circ}\ar@{-}[r]\ar@{-}[l]\ar@/_/@{-}[l]\ar@/^/@{-}[l]&*-<2pt>[o][F-]{ \bullet}\ar@{-}[r]\ar@{-}[d]\ar@/^/@{.}[l]^{\beta'}&\\
&&&&&}}}.
\end{equation}

Next, suppose that each colored vertex occupies the minimal volume $v_0$ predicted by the LQG quantum geometry, then the total quantum volume of the refined 3-sphere is $V=v_0n$ where $n$ is the number of vertices defined by the untilded number operator $N=\Psi^\dag(\psi){\Psi}(\psi)$. For example,  $n=4$ for $\ket{\text{seed}}_c$.
 The Friedman equation for GFT condensate quantum cosmology is obtained by calculating the evolution of $V$ as a scale factor with respect to $\varphi$. Without going into details, we observe that at a refinement level $r$, the volume $V_r$ of the GFT condensate not only sets the scale of the radius $l_A$ of the refined 3-sphere, but also the scale $\varphi$ of cosmic time.
 In this sense, we can interpret the parameter $\beta$ as the inverse ``temperature" of the  apparent horizon, $\beta(r)\propto l_A$, so that
 \begin{equation}\varphi(\beta,r)\propto\varphi(l_A,r)\end{equation}
 depends monotonically on $V_r$. 
\end{example}
Relating the $\beta$ to the scale or the number of vertices in a GFT condensate has another {\it a posteriori} advantage: a fixed $\beta$ means a fixed finite number of Fock excitations in the GFT condensate, so that the tracial states can be directly defined for this condensate state without worrying the types of von Neumann algebras. A point to note is that for a TFD state $\ket{\Omega_S}$, the number of excitation in the untilded system is $\braket{\Omega_S|\phi^\dag\phi|\Omega_S}=\sinh^2\abs{b(\beta)}\equiv n$, which obviously has a  monotonic dependence on $\beta$. Now  consider the ``off-diagonal" correlator $\braket{\Omega_S|\Psi^\dag(\psi)\Psi(\psi')|\Omega_S}$. By the background-independent  definition of the field operators and \eqref{pprime}, we see that the $\psi'$ will be identified with $\psi$ after the Bogoliubov transformation to the unsmeared group fields. As a consequence, we have, by relating $\beta$ to $n$,
\begin{equation}\label{4.14}
\braket{\Omega_S|\Psi^\dag(\psi)\Psi(\psi')|\Omega_S}\sim n,
\end{equation}
which shows the off-diagonal long-range order in $\ket{\Omega_S}$, an essential characterization of BEC \cite{ODLRO}.

\section{Conclusion and discussion}\label{s5}
In this paper, we have studied the TFD extension of GFT in the operator-algebraic formulation with the second-quantization interpretation. In view of the equilibrium Gibbs states for GFT recently obtained in \cite{KO18}, this TFD extension leads to an equilibrium GFT ``at finite temperature", which in turn reveals the ``thermal" aspects of quantum gravity.

The explicit construction takes analogy from the the TFD states of bosonic oscillators, since the GFT state space in the second quantized formulation has the structure of the bosonic Fock space. We have defined the squeezed states \eqref{38} on the Weyl algebra of group fields and shown that these squeezed states can take the standard form of TFD states \eqref{42} on the doubled Hilbert space. 
We have worked in the operator-algebraic approach to TFD where the ``tilde" algebra of observables is the commutant of the original one. In this  formulation we represent the TFD states for GFT as fat graphs with doubled-line edges.

The constructed TFD states  are condensate GFT states ``at finite temperature" when viewed from the original GFT vacuum. With this understanding, we have discussed the possible appearances of the TFD states  in the GFT condensate description of black hole horizon and quantum cosmology.

The algebraic KMS Gibbs state is of course not the only equilibrium state in GFT.  The maximal entropy principle  gives the generalized Gibbs states ($\propto e^{\sum_l\beta_lO_l}$) of GFT \cite{CKO19,Kot19}, which allows a formulation of statistical field theory of quantum polyhedra. The  thermal field theories in this approach are recently investigated in \cite{AK20}. In that approach, the TFD states are directly defined by the thermal Bogoliubov transformations on the doubled Fock states.
 Here we have taken a detour through the von Neumann algebras and return to the local forms the product of two Hilbert spaces. Despite the similar final forms of the TFD states, there is in fact a major difference: In their approach the  Gibbs factors in the TFD states are given by the generalized Gibbs states with an infinite number of ``temperatures'' $\beta_l$s corresponding to the modes of quantum fields; in our approach the symmetry is algebraic so that all field modes have the same CCR algebra and the same global $\beta$-parameter. When localized to the product form, we can still obtain the finite-dimensional cases with different local $\beta$-parameters. In this sense, the difference lies in the global equilibrium configurations: In their approach one can has an
extensive number of conserved charges and the global thermal equilibrium state is the generalized Gibbs state, whereas in our approach we can only have the global thermal equilibrium states as the canonical Gibbs states. In quantum many-body physics, the generalized Gibbs states usually appear as the equilibrium states of the  integrable  systems with an
extensive number of conserved charges, which is, however, not the canonical sense of thermalization. 
In our approach, the canonical Gibbs state is the result of the canonical thermalization. In view towards the semiclassical gravitational thermodynamics in the canonical statistical-mechanical formulation, we think our approach remains physically useful.

\section*{Acknowledgements} 
The author thanks Mehdi Assanioussi, Isha Kotecha and the anonymous referee for various insightful comments. 
This work is partially supported by the National Natural Science Foundation of China through the Grant Nos. 11875006 and
11961131013.
\appendix
\section{Thermo field dynamics extension via Hopf algebra}\label{appA}
In this appendix,  we present another algebraic approach to the TFD extension of GFT via Hopf algebra. The starting point is the Hopf algebra of group fields or of $\mathbb{C}$-functions on group $G$, with which the  GFT on the Drinfeld double $D(G)$ of $G$ has been studied in \cite{Kra07} (see also \cite{BGO11}). Our strategy here is,  instead of considering the Drinfeld double, to apply the quantum-algebra interpretation of TFD doubling \cite{CDDIRV98} to the Hopf algebra of group fields.

\subsection{Hopf algebra of group fields}
The group fields, as functions on groups, can have the structure of a Hopf algebra. To see this, let us first recall that, in the noncommutative metric representation of GFT \cite{BO10}, the convolution of group fields is transformed into the noncommutative $\star$-product of plane wavefunctions $e_g(x)$, e.g.
\begin{align}
\int_\mathbb{R} dx \hat{\phi}_1\star\hat{\bar{\phi}}_2(x)=&\int_\mathbb{R} dx \int_G d\mu(g)\phi_1(g)\phi_2(g^{-1})e_g\star e_{g^{-1}}(x)=\int_\mathbb{R} dx \int_G d\mu(g)\phi_1(g)\phi_2(g^{-1}) e_{gg^{-1}}(x)\nonumber\\
=&\int_G d\mu(g)\phi_1(g)\phi_2(g^{-1})\delta(gg^{-1})=\int_G d\mu(g)\phi_1(g)\phi_2(g)\label{54}
\end{align}
where $\hat{\bar{\phi}}(x)=\hat{{\phi}}(-x)$ and $e_g\star e_{g'}=e_{gg'}$.

 From \eqref{54} we can almost see the structure of a Hopf algebra. Indeed, we have exactly a Hopf algebra $\mathsf{A}^*$ of group fields: The product or multiplication  in $\mathsf{A}^*$ is simply
\begin{equation}
\phi_1\star\phi_2(g)=\phi_1(g)\phi_2(g),\quad g\in G,
\end{equation}
and obviously the unit  is the constant group field $\iota(g)=\phi_I(g)\equiv1$ which is the  identity element of $\mathsf{A}^*$ . The coproduct in $\mathsf{A}^*$ is
\begin{equation}
\Delta^*\phi(g,g')=\phi(gg'),\quad g,g'\in G.
\end{equation}
This is because the function algebra $\mathcal{F}(G\times G)\cong\mathcal{F}(G)\otimes\mathcal{F}(G)$, so that the tensor product in a Hopf algebra is transfered to a two-fold group field. The co-unit in $\mathsf{A}^*$ is then  $\varepsilon(\phi)=\phi(g_0)$ with $g_0$ being the identity element of $G$ such that $\iota\circ\varepsilon(\phi)=\phi_I(g_0)=1$.  The antipode is, as in \eqref{54}, $\mathsf{s}(\phi(g))=\phi(g^{-1})$; the involution is simply the complex conjugation. It is easy to see that the defining relation for a Hopf algebra, after restoring the tensor product notation,
\begin{equation}
\star\circ(1\otimes\mathsf{s})\circ\Delta^*=\star\circ(\mathsf{s}\otimes1)\circ\Delta^*=\iota\circ\varepsilon
\end{equation}
 holds here.

The convolution of group fields \eqref{54} defines an inner product and hence a pairing of group fields, 
\begin{equation}
\braket{\phi_1,\phi_2}=\int_G d\mu(g)\phi_1(g)\phi_2(g).
\end{equation}
For reasons explained in \cite{Kra07}, it is also useful to work with the dual algebra $\mathsf{A}$ of $\mathsf{A}^*$. $\mathsf{A}$ is also a Hopf algebra with the product $\bullet$ given by the duality $\braket{\phi_1\bullet\phi_2,\phi}=\braket{\phi_1\otimes\phi_2,\Delta^*\phi}$, that is,
\begin{equation}
\phi_1\bullet \phi_2(g)=\int_G d\mu(h)\phi_1(h)\phi_2(h^{-1}g),\quad g,h\in G.
\end{equation}
The identity of this $\bullet$-product is then the $\delta$-function $I(g)=\delta_{g_0}(g)$ such that $I\bullet\phi(g)=\phi(g)$. The coproduct and the co-unit in $\mathsf{A}$ are respectively
 \begin{equation}
\Delta\phi(g,g')=\phi(g)\delta_g(g'),\quad \varepsilon(\phi)=\int_G d\mu(g)\phi(g).
\end{equation}
The antipode $\mathsf{s}$ is defined as in $\mathsf{A}^*$, but the involution $*$ is now $\phi^*(g)=\overline{\phi(g^{-1})}$. In this case of $\mathsf{A}$, we still have the relation $\bullet\circ(1\otimes\mathsf{s})\circ\Delta=\bullet\circ(\mathsf{s}\otimes1)\circ\Delta=I\circ\varepsilon$, if we simply put $\Delta\phi=\phi(g)\otimes\phi(g)$ in the tensor product notation.

The $\bullet$-product induces a presentation of the gauge invariance of group fields through the projector $\vartheta=\vartheta\bullet\vartheta$ on $\mathcal{A}^{\otimes n}$ where
\begin{equation}\label{62}
\vartheta(\{g_i\})=\int_{G}d\mu(h)\prod_{i=1}^n\delta_{g_0}(g_i^{-1}h).
\end{equation}
It is easy to see that $\vartheta\bullet\phi(\{hg_i\})=\vartheta\bullet\phi(\{g_i\})$ for $h\in G$. The $\vartheta$ can be considered as a differential, so that the $\delta$-functions in \eqref{62}  are naturally the propagators of a single-fold group field. According to the quantum geometric interpretation of GFT, a gauge-invariant group field generates a quantum polyhedron. Then the Hopf algebra $\mathsf{A}$ dictates the way in which these quantum polyhedra are related. By recalling how the simplical spin foams are generated by the perturbative Feynman diagrams of GFT, the $\mathsf{A}$ is furthermore related to the Hopf algebraic approach to the coarse-graining/renormalization in spin foam models \cite{Mar03}. For example, the $\bullet$-product convolutes two group fields and hence glues two polyhedra; the coproduct $\Delta$ unfolds a spin foam into two subfoams and trivilizes or coarse grains one of them.\footnote{Here the block transforms are encoded in the coproduct, so the antipode in $\mathsf{A}$ is much simpler.}

\subsection{Thermo field dynamics doubling}
Now we show that the TFD can be obtained by the $q$-deformation of the Hopf algebra $\mathsf{A}$. To this end, let us rewrite the coproduct in $\mathsf{A}$ as
\begin{equation}
\Delta:\mathsf{A}\rightarrow\mathsf{A}\otimes\mathsf{A};\quad    \Delta\phi(g)=\phi(g)\delta_g(g')=\frac{1}{2}\bigl(\phi(g)\delta_g(g')\otimes1+1\otimes\phi(g)\delta_{g}(g')\bigr).
\end{equation}
Adopting the second-quantization interpretation of GFT, we see that $\mathsf{A}$ can be made into a bosonic Heisenberg-Weyl algebra by adding a central operator $h=\delta(\{g_i\},\{g_j\})/2$ and the number operator $N=\phi^\dag\phi$. For the purpose of describing TFD doubling, we can focus only on the $q$-deformation of the coproducts of $\phi$:
\begin{equation}
\Delta\phi_q(g)=\frac{1}{2}\bigl(\phi(g)\delta_g(g')\otimes q^{1/2}+q^{-1/2}\otimes\phi(g)\delta_{g}(g')\bigr),
\end{equation}
where the $q$ is required to satisfy $\abs{q}=1$.

The coproduct $\Delta$ naturally doubles the degrees of freedom in the Fock space generated by the group fields $(\phi,\phi^\dagger)$. The following result shows that the  $q$-deformed coproduct can be related to the TFD doubling,
 \begin{proposition}
Let $q=e^{2\theta}$ be the parameter of the $q$-deformation. If the $\theta=\theta(\beta)$ is a squeezing parameter, then the TFD squeezed states in GFT are obtained by the $q$-deformed coproduct of group fields.
\end{proposition}
\begin{proof}
The following proof is an adaptation of \cite{CDDIRV98} to GFT. Consider the ``normalized" $q$-deformed coproduct of group fields and their derivatives,
\begin{align*}
A_q\equiv&\frac{\Delta\phi_q(g)}{\sqrt{[2]_q}}=\frac{1}{2\sqrt{[2]_q}}\bigl(e^{\theta}\phi(g)\delta_g(g')+e^{-\theta}\tilde{\phi}(g)\delta_{g}(g')\bigr),\\
B_q\equiv&\frac{1}{\sqrt{[2]_q}}\frac{\delta(\Delta\phi_q(g))}{\delta\theta}=\frac{1}{2\sqrt{[2]_q}}\bigl(e^{\theta}\phi(g)\delta_g(g')-e^{-\theta}\tilde{\phi}(g)\delta_{g}(g')\bigr).
\end{align*}
Then $A_q+B_q=2e^{\theta}\phi(g')/\sqrt{[2]_q}$ and $A_q-B_q=2e^{-\theta}\tilde{\phi}(g')/\sqrt{[2]_q}$. By introducing 
\begin{align*}
A(\theta)=&\frac{\sqrt{[2]_q}}{2\sqrt{2}}(A_{q(\theta)}+A_{q(-\theta)}-B^\dag_{q(\theta)}+B^\dag_{q(-\theta)}),\\
B(\theta)=&\frac{\sqrt{[2]_q}}{2\sqrt{2}}(B_{q(\theta)}+B_{q(-\theta)}-A^\dag_{q(\theta)}+A^\dag_{q(-\theta)}),
\end{align*}
we have for any $g'\in G$,
\begin{equation}\label{65}
\phi(\theta)=\frac{1}{\sqrt{2}}(A(\theta)+B(\theta))=\phi\cosh\theta-\tilde{\phi}^\dag\sinh\theta
\end{equation}
and likewise $\tilde{\phi}(\theta)=(A(\theta)-B(\theta))/\sqrt{2}$. Thus, \eqref{65} matches the Bogoliubov transformation for TFD squeezed vacuum $\ket{\Omega_S}$ such that $\phi(\theta)\ket{\Omega_S}=0$.
\end{proof}

Now that the TFD states can be obtained from the $q$-deformed Hopf algebra of group fields, it is tempting to relate the TFD states to the hyperbolic quantum geometry of $q$-deformed LQG \cite{DG13}. However, in $q$-deformed LQG  the gauge group is deformed into a quantum group, for example, $U_q(\mathfrak{su}(2))$, whereas in the TFD extension of GFT, what is deformed is the Hopf algebra of group fields which are still defined on the original gauge group. In other words, the $q$-deformed gauge group changes the algebraic data on, or the coloring of, the  spin network states of the undeformed LQG, while in the TFD case the geometric coloring is preserved but the spin network states are squeezed into a particular condensate state. Physically, the $q$-deformation of the gauge group puts a ``cutoff", related to the cosmological constant, on the LQG vertex amplitudes, but the TFD vacuum state mixes the spin network states of LQG. In this sense, there does not seem to be a direct relation between these two quantum deformations. 

From the point of view of coarse-graining  spin foams, the $q$-deformed coproduct does not completely coarse grain the spin foams, but retains the mixing of all the subfoams.

\subsection{A speculative quantum-geometric  interpretation}
In the algebraic formulation of TFD in the main text,  $\tilde{\mathcal{M}}=\mathcal{M}'\neq\mathcal{M}$ in general. In other words, the ``tilde" system has an observable algebra different from that of the original GFT system (of  LQG), and they should represent different but correlated physical systems.

We  observe the following: (i) The $\mathcal{M}'$ is {\it emergent} in the sense that it is defined according to the the algebraic TFD rules and the equilibrium GFT Gibbs states. This is similar to the inclusion of matter into quantum gravity by generalizing GFTs to those defined on the Drinfeld double $D(G)$ of the gauge group $G$ \cite{Kra07}, where the dual group of $G$ encodes the matter content. (ii) To each GFT Fock state $\ket{\{g_i\},n}$ there corresponds a ``tilde" state $\ket{\{g'_i\},\tilde{n}}$. In a sense, this means that every piece of quantum geometric data is coupled to some other state, as is proposed for gravity coupled with scalar matter in \cite{XM09}. (iii) The  ``tilde" states are nontrivial only when there are interactions, because for the free Fock states $\mathcal{M}'_F=\eta{\bf1}$ is trivial.
 Since these ``tilde" states  are in the commutant  $\mathcal{M}'$ they do not affect the dynamics of $\mathcal{M}$.
  Summarizing the above observations, we speculate that  the TFD states in GFT are {\it quantum polyhedra dressed with the emergent ``matter" states}.
 
 Suppose $\mathcal{M}=\pi_\omega^{\prime\prime}(\mathcal{W})=\pi'(K)$ where $\pi(K)$ is  a representation of a group $K$, then $\mathcal{M}'=\pi^{\prime\prime}(K)$ is also a von Neumann algebra, so that $\pi^{\prime\prime}(K)=\pi(K)$. In other words, $\mathcal{M}'$ is determined by the (irreducible) representations of the group $K$, which is similar to how matter quantum fields are classified by the irreducible representation of the Lorentz group. However, we know nothing about the group $K$  at present.

\bibliographystyle{amsalpha}

\end{document}